\PassOptionsToPackage{pdfpagelabels=false}{hyperref} 
\documentclass[twoside]{aiml12}
 
\usepackage{aiml12macro}

\usepackage{graphicx}
\usepackage{amsmath}
\usepackage{amssymb}
\usepackage{mathtools}
\mathtoolsset{centercolon}

\def\lastname{G\"oller, Meier, Mundhenk, Schneider, Thomas, Wei\ss}
\usepackage{hl-sat,stackrel}

\newcommand{\nonelementary}{non-elementary\xspace}

\newcommand{\Reportout}[1]{#1}
\newcommand{\Reportin}[1]{}
\renewcommand{\Reportout}[1]{}
\renewcommand{\Reportin}[1]{#1}

\newcommand{\shtitle}{The Complexity of Monotone Hybrid Logics over
Lin.\ Frames and the Nat.\ Numbers}
\shttltrue

\begin{document}

\begin{frontmatter}
\title{The Complexity of Monotone Hybrid Logics over Linear Frames and the Natural Numbers}
\author{Stefan G\"oller}
\address{Department of Computer Science, Universit\"at Bremen, Germany} 
\author{Arne Meier}
\address{Institute of Theoretical Computer Science, Leibniz Universit\"at Hannover, Germany} 
\author{Martin Mundhenk}
\address{Institute of Computer Science, Friedrich-Schiller-Universit\"at Jena, Germany}
\author{Thomas Schneider}
\address{Department of Computer Science, Universit\"at Bremen, Germany} 
\author{Michael Thomas}
\address{TWT GmbH, Germany}
\author{Felix Wei\ss}
\address{Institute of Computer Science, Friedrich-Schiller-Universit\"at Jena, Germany}


  
  \begin{abstract}
 Hybrid logic with binders is an expressive specification language.
 Its satisfiability problem is undecidable in general.
 If frames are restricted to \Nat or general linear orders,
 then satisfiability is known to be decidable, but of non-elementary
complexity.
 In this paper, we consider monotone hybrid logics (i.e., the Boolean
connectives are conjunction and disjunction only)
 over \Nat and general linear orders.
 We show that the satisfiability problem remains non-elementary over
linear orders,
 but its complexity drops to \PSPACE-completeness over \Nat.
 We categorize the strict fragments arising from different combinations
 of modal and hybrid operators into \NP-complete and tractable (i.e.
complete for $\NC1$ or \LOGSPACE).
 Interestingly, \NP-completeness depends only on the fragment and not
on the frame.
 For the cases above \NP, satisfiability over linear orders is harder
than over \Nat,
 while below \NP it is at most as hard.
 In addition 
 we examine model-theoretic properties
 of the fragments in question.
  \end{abstract}

 \end{frontmatter}

\section{Introduction}

Hybrid logic is an extension of modal logic with nominals, satisfaction operators and binders.
The downarrow binder $\dna$, which is related to the freeze operator in temporal logic \cite{hen90}, 
provides high expressivity. 
The price paid is the undecidability of the satisfiability problem
for the hybrid language with the downarrow binder $\dna$ \cite{blse95,gor96,ArBM99b}. 
In contrast, modal logic, and its extension with nominals and the satisfaction operator,
is \PSPACE-complete \cite{lad77,ArBM99b}.

In order to regain decidability, syntactic and semantic restrictions have been considered.
It has been shown in \cite{cafr05} that the absence of certain combinations of universal operators ($\Box$, $\wedge$)
with $\dna$ brings back decidability, and that the hybrid language with $\dna$ is decidable over
frames of bounded width. 
Furthermore, this language is decidable over transitive and complete frames
\cite{mssw05}, and over frames with an equivalence relation (ER frames) \cite{musc09}.
Adding the at-operator $\at$---which allows to jump to states named by nominals---leads to undecidability over transitive frames \cite{mssw05},
but not over ER frames \cite{musc09}.
Over linear frames and transitive trees, $\dna$ on its own does not add expressivity, but combinations with
$\at$ or the global modality---an additional $\Diamond$ interpreted over the universal relation---do. 
These languages are decidable and of non-elementary complexity \cite{frrisc03,mssw05};
if the number of state variables is bounded, then they are of elementary complexity \cite{SchwentickW07,Weber07,BozzelliL08}.

We aim for a more fine-grained distinction between fragments of different complexities
by systematically restricting the set of Boolean connectives and combining this with
restrictions to the modal/hybrid operators and to the underlying frames. 
In \cite{MMSTWW09}, we have focussed on four frame classes that allow cycles, 
and studied the complexity of satisfiability for fragments obtained by arbitrary combinations
of Boolean connectives and four modal/hybrid operators.
The main open question in \cite{MMSTWW09} is the one for tight upper bounds for monotone fragments including the $\Box$-operator.
Even though there are many logics for which 
the restriction to monotone Boolean connectives leads to a significant decrease in complexity,
it is not straightforward, and therefore interesting to find out, 
where this happens for hybrid logics.

In this study, we classify the computational complexity of satisfiability
for monotone fragments of hybrid logic with arbitrary combinations of the
operators $\Diamond$, $\Box$, $\dna$ and $\at$ over linear orders and the natural numbers.
Whereas the full logic is non-elementary and decidable \cite{mssw05} for both frame classes,
we show that in the monotone case this high complexity is gained only over linear orders
and drops to \PSPACE-completeness over the natural numbers.
Informally speaking, the reason is that linearly ordered frames
may consist of arbitrarily many dense parts that can be distinguished using the expressive power
of all four operators.
These dense parts and their distances are used to store information
that cannot be stored in a frame without dense parts as, e.g., the natural numbers.
For all other monotone fragments that contain the $\Diamond$-operator, we show \NP-completeness 
independent on the frame class, for linear orders, all remaining fragments (i.e. the fragments without $\Diamond$)
can be shown to be $\NC1$-complete.
The reason is, informally speaking, that all (sub-)formulas of the form $\Box\alpha$
are easily satisfied in a state without successor, 
which can essentially be used to reduce this problem to the satisfiability problem for monotone propositional formulae.
This argument does not go through over the natural numbers, a total frame where every state has a successor.
Over this frame class, we give a decision procedure that runs in logarithmic space for
the fragment with all operators except $\Diamond$ (and prove a matching lower bound),
and in $\NC1$ for all other fragments.

These results give rise to two interesting observations.
First, the \NP-completeness results are independent on the frame class.
Second, for the fragment whose satisfiability problem is above \NP, linear orders make the problem harder than the natural numbers,
and for the richest fragment below \NP, it is the opposite way round---the natural numbers make the problem harder than linear orders. 
Notice also that, in the case where Boolean operators are not restricted to monotone ones, all fragments are \NP-hard.


Our results are shown in Figure \ref{fig:m_Verband}.

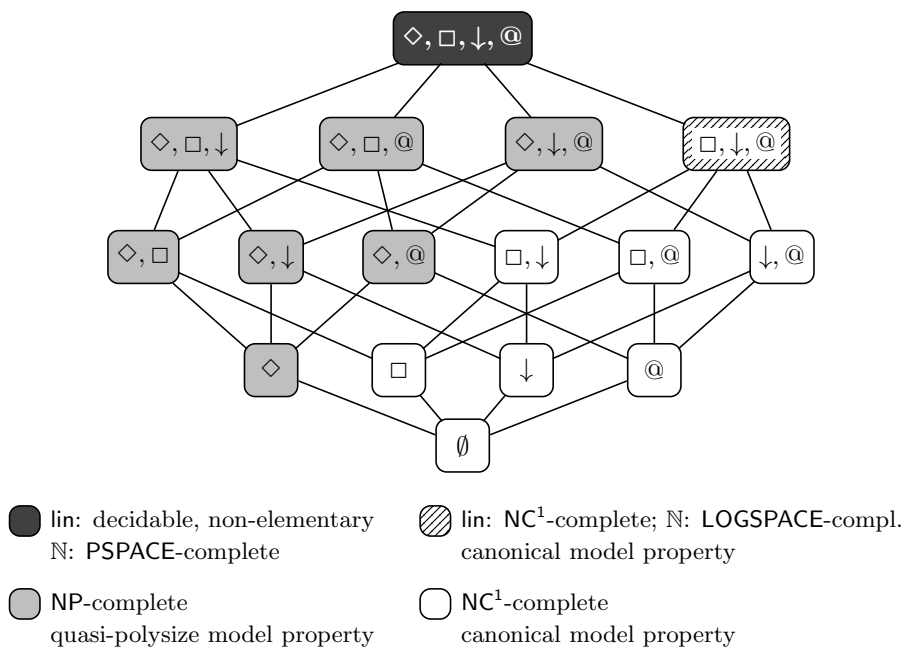
\begin{figure}[h]
        \centering
                \begin{tikzpicture}[ -,y=1cm,
                                     x=1.2cm,
                                     auto,
                                     node distance=1.8cm,
                                     semithick,
                                     rounded corners,
                                     pspace/.style=  {style=rectangle,
                                                      draw=black,
                                                      fill=darkgray,
                                                      text=white,
                                                      minimum size=7mm},
                                     pspaceLeg/.style=  {style=rectangle,
                                                      draw=black,
                                                      minimum size=4mm,
                                                      fill=darkgray},
                                     np/.style=      {style=rectangle,
                                                      draw=black,
                                                      minimum size=7mm,
                                                      fill=lightgray},
                                     npLeg/.style=    {style=rectangle,
                                                      draw=black,
                                                      minimum size=4mm,
                                                      fill=lightgray},
                                     nc1/.style=     {style=rectangle,
                                                      draw=black,
                                                      minimum size=7mm,
                                                      fill=white},
                                     nc1Leg/.style=  {style=rectangle,
                                                      draw=black,
                                                      minimum size=4mm,
                                                      fill=white},
                                     logspace/.style= {style=rectangle,
                                                      draw=black,
                                                      minimum size=7mm,
                                                      fill=white,
                                                      postaction={pattern=north east lines}},
                                     logspaceLeg/.style= {style=rectangle,
                                                      draw=black,
                                                      minimum size=4mm,
                                                      fill=white,
                                                      postaction={pattern=north east lines}}
                                   ]
  \node[nc1]    (empty)    at (-3,0)                          {$\emptyset$};

  \node[nc1]   (bo)   at (-3.7,1)	{$\Box$};
  \node[np]     (di)   at (-5.1,1)	{$\Diamond$};
  \node[nc1]    (da)   at (-2.3,1)		{$\dna$};
  \node[nc1]    (at)   at (-.9,1)		{$\at$};

  \node[np]     (dibo)  at (-6.5,2.5)	{$\Diamond,\Box$};
  \node[np]     (dida)  at (-5.1,2.5)	{$\Diamond,\dna$};
  \node[np]     (diat)  at (-3.7,2.5)	{$\Diamond,\at$};
  \node[nc1]   (boda)  at (-2.3,2.5)	{$\Box,\dna$};
  \node[nc1]   (boat)  at (-.9,2.5)	{$\Box,\at$};
  \node[nc1]    (daat)  at (.5,2.5)	{$\dna,\at$};

  \node[np]     (diboda)	at (-6,4)	{$\Diamond, \Box, \dna$};
  \node[np]     (diboat)  	at (-4,4)	{$\Diamond, \Box, \at$};
  \node[np]     (didaat)  	at (-2,4)	{$\Diamond, \dna, \at$};
  \node[logspace]   (bodaat)  	at (0,4)	{\MyColorBox[white]{$\Box, \dna, \at$}};  
  
  \node[pspace] (all)	at (-3,5.4)	{$\boldsymbol\Diamond\boldsymbol,\,\boldsymbol\Box\boldsymbol,\, \boldsymbol\dna\boldsymbol,\, \boldsymbol\at$};  
  
  \node[pspaceLeg,
			label={0:\begin{tabular}{@{}l@{}}
								\\\small\strut $\lin$: decidable, non-elementary\\
								\small\strut $\Nat$: \PSPACE-complete
							 \end{tabular}}]   (pspace)   at (-7.8, -1)    {};

  \node[npLeg,label={0:\begin{tabular}{@{}l@{}}
								\\\small\strut\NP-complete\\
								\small\strut quasi-polysize model property
							 \end{tabular}}]       (np)       at (-7.8, -2.1)    {};

  \node[logspaceLeg,label={0:\begin{tabular}{@{}l@{}}
								\\\small\strut $\lin$: \NC1-complete; 
								\small\strut $\Nat$: \L-compl.\\
								\small\strut canonical model property
							 \end{tabular}}] (logspace) at (-3.3, -1)    {};

  \node[nc1Leg,label={0:\begin{tabular}{@{}l@{}}
								\\\small\strut\NC1-complete\\
								\small\strut canonical model property
							 \end{tabular}}]      (nc1)      at (-3.3, -2.1)    {};

  \path[]
	(empty)	edge		(di)
					edge	[shorten >=-.04cm,shorten <=-.05cm]	(bo)
					edge[shorten >=-.04cm,shorten <=-.05cm]		(da)
					edge		(at)
	(di)			edge		(dida)
					edge[shorten >=-.02cm,shorten <=-.06cm]			(dibo)
					edge[shorten >=-.02cm,shorten <=-.06cm]			(diat)
	(bo)			edge	[shorten >=-.02cm,shorten <=-.02cm]		(dibo)
					edge	[shorten >=-.05cm,shorten <=-.06cm]		(boda)
					edge	[shorten >=-.02cm,shorten <=-.02cm]		(boat)
	(da)	 		edge[shorten >=-.02cm,shorten <=-.02cm]			(dida)
					edge			(boda)
					edge[shorten >=-.02cm,shorten <=-.02cm]			(daat)
	(at)			edge[shorten >=-.02cm,shorten <=-.02cm]			(diat)
					edge          (boat)
					edge[shorten >=-.06cm,shorten <=-.06cm]			(daat)
         (dibo) edge          (diboda)
                edge[shorten >=-.07cm,shorten <=-.02cm] (diboat)
         (dida) edge          (diboda)
                edge[shorten >=-.03cm,shorten <=-.02cm](didaat)
         (diat) edge          (diboat)
                edge[shorten >=-.02cm,shorten <=-.06cm]          (didaat)
         (boda) edge[shorten >=-.02cm,shorten <=-.02cm](diboda)
                edge[shorten >=-.04cm,shorten <=-.02cm](bodaat)
         (boat) edge[shorten >=-.02cm,shorten <=-.02cm]          (diboat)
                edge          (bodaat)
         (daat) edge[shorten >=-.05cm,shorten <=-.02cm](didaat)
                edge          (bodaat)
         (diboda)[shorten >=-.08cm,shorten <=-.03cm]edge            (all)
         (diboat)[shorten >=-.04cm,shorten <=-.02cm]edge            (all)
         (didaat)                edge            (all)
         (bodaat)[shorten >=-.08cm,shorten <=-.04cm]edge            (all);
\end{tikzpicture}
          \caption{%
            Our complexity results for satisfiability over
            linear frames $(\lin)$ and the natural numbers $(\Nat)$
            for hybrid logic
            with monotone Boolean operators and different combinations of modal/hybrid operators%
          }
          \label{fig:m_Verband}
\end{figure}


\section{Preliminaries}
\label{sec:prelims}

  \ourparagraph{Hybrid Logic.}
    In the following, we introduce the notions and definitions of hybrid logic. The terminology is largely taken from \cite{arblma00}.

    Let $\PROP$ be a countable set of \emph{atomic propositions}, $\NOM$ be a countable set of \emph{nominals}, $\SVAR$ be a countable set of \emph{variables} and $\ATOM = \PROP \cup \NOM \cup \SVAR$.
    We adhere to the common practice of denoting atomic propositions by $p,q,\ldots$, nominals by $i,j,\ldots$, and variables by $x,y,\ldots$
    We define the language of \emph{hybrid (modal) logic} $\HL$ as the set of well-formed formulae of the form
      \[
        \varphi \bnf a \mid \top \mid \bot \mid \neg\varphi \mid \varphi\land\varphi \mid \varphi\lor\varphi \mid \Diamond \varphi \mid \Box \varphi \mid \dna x.\varphi  \mid \at_t \varphi
      \]
    where $a \in \ATOM$, $x \in \SVAR$ and $t \in \NOM \cup \SVAR$.

We define the usual Kripke semantics only to be able to refer to already existing
results. We will then simplify the standard semantics for monotone formulae.
Formulae of $\HL$ are interpreted on \emph{(hybrid) Kripke structures} $K=(W,R,\eta)$, 
consisting of a set of \emph{states} $W$, a \emph{transition relation} $R\colon W\times W$, 
and a \emph{labeling function} $\eta\colon \PROP\cup\NOM \to\wp(W)$ that maps $\PROP$ and $\NOM$ 
to subsets of $W$ with $|\eta(i)| = 1$ for all $i \in \NOM$. 
The relational structure $(W,R)$ is the \emph{Kripke frame} underlying $K$.
In order to evaluate $\dna$-formulae, an assignment $g\colon \SVAR \to W$ is necessary. 
Given an assignment $g$, a state variable $x$ and a state $w$, an \emph{$x$-variant $g^x_w$ of $g$} is defined
by $g^x_w(x)=w$ and $g^x_w(x')=g(x')$ for all $x \neq x'$. 
For any $a \in \ATOM$, let $[\eta,g](a)=\{g(a)\}$ if $a \in \SVAR$ and $[\eta,g](a)=\eta(a)$, otherwise.
    The satisfaction relation of hybrid formulae is defined as follows.

    \begin{tabbing}
    \hspace*{0pt}\= $K,g,w \models \varphi\land\psi$~ \= ~~if and only if~~ \= $\exists w' \in W(wRw' \mathbin{\&} K,g,w' \models \varphi)$\\[2px]
      \> $K,g,w \models a$ \> ~~if and only if \> $w \in [\eta,g](a)$, $a \in \ATOM$, \\[2px]
      \> $K,g,w \models \top$, \> ~~and $K,g,w \not\models \bot$, \>  \\[2px]
      \> $K,g,w \models \neg\varphi$ \> ~~if and only if \> $K,g,w \not\models \varphi$, \\[2px]
      \> $K,g,w \models \varphi\land\psi$ \> ~~if and only if \> $K,g,w \models \varphi$ and $K,g,w \models \psi$, \\[2px]
      \> $K,g,w \models \varphi\lor\psi$ \> ~~if and only if \> $K,g,w \models \varphi$ or $K,g,w \models \psi$, \\[2px]
      \> $K,g,w \models \Diamond \varphi$ \> ~~if and only if \>$\exists w' \in W(wRw' \mathbin{\&} K,g,w' \models \varphi)$,  \\[2px]
      \> $K,g,w \models \Box \varphi$ \> ~~if and only if \> $\forall w' \in W(wRw' \Rightarrow K,g,w' \models \varphi)$, \\[2px]
      \> $K,g,w \models \at_t \varphi$ \> ~~if and only if \> $K,g,[\eta,g](t) \models \varphi$, \\[2px]
      \> $K,g,w \models \dna x. \varphi$ \> ~~if and only if \> $K,g^x_w,w \models \varphi$. 
    \end{tabbing}

    A hybrid formula $\varphi$ is said to be \emph{satisfiable} if there exists a Kripke structure $K=(W,R,\eta)$, a $w \in W$ and an assignment $g\colon \SVAR \to W$ with $K,g,w \models \varphi$.

    The \emph{at} operator $\at_t$ shifts evaluation to the state named by $t\in \NOM \cup \SVAR$.
    The \emph{downarrow binder} $\dna x.$ binds the state variable $x$ to the current state. 
    The symbols $\at_x$, $\dna x.$  
    are called \emph{hybrid operators} whereas the symbols $\Diamond$ and $\Box$ are called \emph{modal operators}. 

    The scope of an occurrence of the binder $\dna$ is defined as usual. 
    For a state variable $x$, an occurrence of $x$ or $\at_x$ in a formula $\varphi$ is called \emph{bound}
    if this occurrence is in the scope of some $\dna$ in $\varphi$, \emph{free} otherwise.
    $\varphi$ is said to contain a free state variable if some $x$ or $\at_x$ occurs free in $\varphi$.

    Given two formulae $\varphi,\alpha$ and a subformula $\psi$ of $\varphi$, we use $\varphi[\psi/\alpha]$
    to denote the result of replacing each occurrence of $\psi$ in $\varphi$ with $\alpha$.
    For considering fragments of hybrid logics, we define subsets of the language $\HL$ as follows.
    Let $O$ be a set of hybrid and modal operators, i.e., a subset of $\{\Diamond,\Box,\dna,\at\}$.
    We define $\HL(O)$ to denote the set of well-formed hybrid formulae using only the operators in $O$,
    and $\MHL(O)$ to be the set of all formulae in $\HL(O)$ that do not use $\neg$.

  \ourparagraph{Properties of Frames.}
    A \emph{frame} $F$ is a pair $(W,R)$, where $W$ is a set of states and $R\subseteq W\times W$ a transition relation.
    A frame $F = (W,R)$ is called
    \begin{itemize}
      \parsep0pt
      \itemsep0pt
      \item
        transitive if $R$ is transitive (for all $u,v,w \in W$: $uRv \land vRw \limplies uRw$),
      \item
        linear if $R$ is transitive, irreflexive and trichotomous ($\forall u,v \in W$: $uRv$ or $u=v$ or $vRu$),
    \end{itemize}

    In this paper we consider the class of all linear frames, denoted by $\lin$,
    and the singleton frame class $\{(\Nat,<)\}$, denoted by \Nat.
    Obviously, $\Nat \subseteq \lin$.

  \ourparagraph{Notational convenience.}
    We can make some simplifying assumptions about syntax and semantics,
    of $\HL(O)$ and $\MHL(O)$,
    which do not restrict generality.
    (1) If $\dna \in O$, then formulae do not contain any nominals. Those can be simulated by free state variables.
    (2) Free state variables are never bound later in the formula,
    and every state variable is bound at most once.
    The latter is no significant restriction because variables bound multiple times
    can be named apart, which is a well-established and computationally easy procedure.
    (3) Monotone formulae do not contain any atomic propositions.
    This restriction is correct because every monotone formula $\varphi$
    is satisfiable if and only if $\varphi$ with all atomic propositions replaced by $\top$
    is satisfiable. This justifies the following restrictions.
    (4) For binder-free fragments, the domain of the labelling function $\eta$ is restricted to nominals,
    and we re-define $\eta\colon \NOM \to W$.
    Furthermore, the absence of $\dna$ makes assignments superfluous: we write $F,w \models \varphi$ instead of $F,g,w \models \varphi$.
    (5) For binder fragments, the satisfaction relation $\models$ is restricted to Kripke \emph{frames}
    $F=(W,<)$, where $<$ is a linear order, and assignments $g : \SVAR \to W$,
    i.e., we write $F,g,w \models \varphi$.
    (6) Over \Nat, we omit the single Kripke frame,
    i.e., we write $\eta,i \models \varphi$ with $\eta : \NOM \to \Nat$ and $i \in \Nat$ for binder-free fragments,
    and $g,i \models \varphi$ with $g: \SVAR \to \Nat$ for binder fragments.

  \ourparagraph{Satisfiability Problems.}
    The \emph{satisfiability problem} for $\HL(O)$ over the frame class $\Fclass{F}$ is defined as follows:
    \problemdef{%
      $\SaT[\Fclass{F}](O)$%
    }{%
      an $\HL(O)$-formula $\varphi$ (without nominals, see above)%
    }{%
      Is there a Kripke structure $K$ based on a frame $(W,R) \in \Fclass{F}$, 
      an assignment $g\colon \SVAR \to W$ and a $w\in W$ such that $K,g,w\models\varphi$\,?%
    }

    The \emph{monotone satisfiability problem} for $\MHL(O)$ over the frame class $\Fclass{F}$ is defined as follows:
    \problemdef{%
      $\MSaT[\Fclass{F}](O)$%
    }{%
      an $\MHL(O)$-formula $\varphi$ without nominals and atomic propositions%
    }{%
      Is there a Kripke frame $(W,R) \in \Fclass{F}$, 
      an assignment $g\colon \SVAR \to W$ and a $w\in W$ such that $F,g,w\models\varphi$\,?%
    }

    If $\Fclass{F}$ is the class of all frames, we simply write $\SaT(O)$ or $\MSaT(O)$.
    Furthermore, we often omit the set parentheses when giving $O$ explicitly, e.g., $\SaT(\Diamond,\Box,\dna,\at)$.

  \ourparagraph{Complexity Theory.}
    We assume familiarity with the standard notions of complexity theory as, e.\,g., defined in \cite{pap94}.
    In particular, we make use of the classes $\L$, $\NLOGSPACE$, $\NP$, $\PSPACE$, and $\coRE$.
    The complexity class $\NONEL$ is the set of all languages $A$ that are decidable
    and for which there exists no $k\in\N$ such that $A$ can be decided using an algorithm 
    whose running time is bounded by $\exp_k(n)$, where $\exp_k(n)$ is the
    $k$-th iteration of the exponential function (e.g., $\exp_3(n)=2^{2^{2^n}}$).

    Furthermore, we need two non-standard complexity classes whose
    definition relies on circuit complexity and formal languages,
    see for instance \cite{vol99,mah07}.
    The class $\NC{1}$ is defined as the set of languages recognizable by a logtime-uniform family of Boolean
    circuits of logarithmic depth and polynomial size over $\{\land,\lor,\neg\}$,
    where the fan-in of $\land$ and $\lor$ gates is fixed to $2$.
    The class \detLOGCFL is defined as the set of languages reducible in logarithmic space
    to some deterministic context-free language.

    The following relations between the considered complexity classes are known.
		\begin{center}
    	$\NC1 \subseteq \LOGSPACE \subseteq \detLOGCFL \subseteq \NP \subseteq \PSPACE \subset \coRE$.
		\end{center}
    It is unknown whether \detLOGCFL contains \NLOGSPACE or vice versa.

A language $A$ is \emph{constant-depth reducible} to $D$, $A\leqcd D$, if there is a logtime-uniform $\AC{0}$-circuit family with oracle gates for $D$ that decides membership in $A$. Unless otherwise stated, all reductions in this paper are \leqcd-reductions.

  \ourparagraph{Known results.}
    The following theorem summarizes results for hybrid languages with Boolean operators $\wedge,\vee,\neg$
    that are known from the literature. 
    Since $\Box \varphi \equiv \neg \Diamond \neg \varphi$, the $\Box$-operator is implicitly present in all fragments containing $\Diamond$ and negation.

    \begin{theorem}[\cite{ArBM99b,arblma00,bhss05b,frrisc03,mssw05}]
      \label{theo:known_for_and_or_neg}
      ~\par
      \begin{Enum}
        \item
          $\SaT(\Diamond,\dna,\at)$ and $\SaT(\Diamond,\dna)$ are \coRE-complete. \cite{ArBM99b}
        \item 
          $\MSaT(\Diamond,\Box)$ is \PSPACE-hard. \cite{bhss05b}
        \item\label{nonelem}
          $\SaT[\mathfrak{F}](\Diamond,\dna,\at)$,
          for $\mathfrak{F} \in \{\lin,\Nat\}$,
          are in \textup{\textsf{NON\-ELEMENTARY}}. \cite{frrisc03,mssw05}
        \item
          \label{it:known_for_and_or_neg_over_lin_N}
          $\SaT[\mathfrak{F}](\Diamond,\dna)$, $\SaT[\mathfrak{F}](\Diamond,\at)$ and $\SaT[\mathfrak{F}](\Diamond)$,
          with $\mathfrak{F} \in \{\lin,\Nat\}$,
          are \NP-complete. \cite{arblma00,frrisc03}
      \end{Enum}
    \end{theorem}

  \ourparagraph{Our contribution.}
    In this paper, we consider the monotone satisfiability problems $\MSaT[\mathfrak{F}](O)$
    for $\mathfrak{F} \in \{\lin,\Nat\}$ and all $O \subseteq \{\Diamond,\Box,\dna,\at\}$.

\noindent

\section{The hard cases: Non-elementary and \texorpdfstring{$\PSPACE$ results}{PSPACE results}}
\noindent
The hardest cases are those with the complete set of operators.
In the non-monotone case, both satisfiability problems are \nonelementary and decidable \cite{mssw05}.
We show that in the monotone case even this hardness is reached,
but only on linear frames, i.e.~$\MSaT[\lin](\Diamond,\Box,\dna,\at)$ is non-elementary and decidable.
In contrast, on the natural numbers the complexity decreases,
i.e.~we show that $\MSaT[\Nat](\Diamond,\Box,\dna,\at)$ is \PSPACE-complete.

Our proofs use reductions to and from fragments of first-order logic on the natural numbers.
Let $\FOL(<,P)$ be the set of all first-order formulae
that use $<$ as the unique binary relation symbol, and $P$ as the unique unary relation symbol.\footnote{%
I.e. $\FOL(<,P)$ is defined as set of all formulae $\varphi$ as follows.
\[
          \varphi \bnf \top \mid x < y \mid P(x) \mid \neg\varphi \mid \varphi\land\varphi \mid \varphi\lor\varphi \mid
                       \exists x \, \varphi \mid \forall x \, \varphi
\]
for variable symbols $x,y\in\SVAR$.}
Let $\SaT[\Nat]_{\FOL}(<,P)$ denote the set of formulae from $\FOL(<,P)$
which are satisfied by a model that has \Nat as its universe, interprets $<$ as the less-than relation on $\Nat\times\Nat$,
and has an arbitrary interpretation for the predicate symbol $P$.
It was shown by Stockmeyer \cite{stmephdthesis} that $\SaT[\Nat]_{\FOL}(<,P)$ is \nonelementary.

Let $\FOL(<)$ be the fragment of $\FOL(<,P)$ in which the predicate symbol $P$ is not used.
Accordingly, $\SaT[\Nat]_{\FOL}(<)$ denotes the set of formulae that are satisfiable over $\Nat$ and the natural interpretation of $<$.
It was shown by Ferrante and Rackoff \cite{fera79}  that $\SaT[\Nat]_{\FOL}(<)$ is in \PSPACE.

Notice that in both fragments $x=y$ can be expressed as $\neg (x<y\, \vee\, y<x)$.
Moreover, every $n\in\Nat$ can be expressed by $x_n$ in the formula 
$\exists x_0 \cdots \exists x_{n-1} [(\bigwedge_{i=0,1,\ldots,n-1} x_i < x_{i+1}) \wedge \forall y (x_n<y \vee \bigvee_{i=0,1,\ldots,n} y=x_i)]$.

\begin{theorem}
  \label{thm:lin-DBda-nonel}
  $\MSaT[\lin](\Diamond,\Box,\dna,\at)$ is \nonelementary and decidable.
\end{theorem}

\newcommand{\positive}{\macro{pos}}
\newcommand{\negative}{\macro{neg}}

\newcommand{\directSuc}{\macro{dirSuc}}
\newcommand{\directPred}{\macro{dirPred}}
\newcommand{\noDirectPred}{\macro{noDirPred}}

\begin{proof}
  Decidability follows from Theorem \ref{theo:known_for_and_or_neg} \ref{nonelem}.
  To establish \nonelementary complexity,
  we give a reduction from $\SaT[\Nat]_{\FOL}(<,P)$.
  
  We first show how to encode the intepretation of a predicate symbol, represented by a set $P\subseteq\Nat$, in a linear frame $F = (W,<)$ --
  without using atomic propositions and nominals as agreed in Section \ref{sec:prelims}.
  Using free state variables, we can only distinguish linearly many states at any given time. 
  We therefore use finite intervals (finite subchains of $(W,<)$) to encode whether $n\in P$.
  Such an interval---we call it a \emph{marker}---has length 2 (resp.~3) if for the corresponding $n$ holds $n\not\in P$  (resp.~$n\in P$). 
  Accordingly, we call a marker of length 2 (resp.~3) \emph{negative} (resp.~\emph{positive}). 
  These finite intervals are separated by dense intervals---those are intervals wherein every two states have an
  intermediate state, e.g., $[0,1]_{\mathbb{Q}} = \{q \in \mathbb{Q} \mid 0 \leqslant q \leqslant 1\}$.
  For example, the set $P$ with $0,2\not\in P$ and $1\in P$ is represented by the chain in Figure \ref{fig:appthm2}.
  \begin{figure}[ht]
    \begin{center}
      \begin{tikzpicture}[
 scale=0.9,
 -,
 auto,
 node distance=5cm,
 label distance=3mm,
 semithick,
 state/.style={style=circle, draw=black, minimum size=4mm, inner sep=0mm},
 txt/.style={style=rectangle}]
 
	\node[state] (w1) at (0,0) {};
	\node[state] (w2) at (1.5,0) {};
	\path (w1) edge[->,thick] (w2);
	\node[txt] at (0.75,-0.7) {$0\notin P$};
	\node[state] (w3) at (3.5,0) {};
	\draw [-,thick,decorate,decoration={snake,amplitude=.4mm,segment length=2mm}] (w2) -- (w3);
	\node[state] (w4) at (5,0) {};
	\node[txt] at (5,-0.7) {$1\in P$};
	\node[state] (w5) at (6.5,0) {};
	\path (w3) edge[->,thick] (w4) (w4) edge[->,thick] (w5);
	\node[state] (w6) at (8.5,0) {};
	\draw [-,thick,decorate,decoration={snake,amplitude=.4mm,segment length=2mm}] (w5) -- (w6);
	\node[state] (w7) at (10,0) {};
	\path (w6) edge[->,thick] (w7);
	\node[txt] at (9.25,-0.7) {$2 \notin P$};
	\draw [-,thick,decorate,decoration={snake,amplitude=.4mm,segment length=2mm}] (w7) -- (12,0);	
	\node[txt] at (12.5,0) {$\cdots$};
	
	\node[txt] at (.5,-1.5) {Legend:};
	\node[state] (v1) at (1.5,-1.5) {$w$};
	\node[state,label={0:: ~$v$ is a \underline{direct} successor of $w$}] (v2) at (3,-1.5) {$v$};
	\path (v1) edge[->,thick] (v2);
	\node[state] (v3) at (1.5,-2.1) {$w$};
	\node[state,label={0:: ~$w$  and $v$ are begin and end of a \underline{dense} interval}] (v4) at (3,-2.1) {$v$};
	\draw [-,thick,decorate,decoration={snake,amplitude=.4mm,segment length=2mm}] (v3) -- (v4);	
	\node[state] (v5) at (1.5,-2.7) {$w$};
	\draw [-,thick,decorate,decoration={snake,amplitude=.4mm,segment length=2mm}] (v5) -- (2.5,-2.7);	
	\node[txt,label={0:\hspace{-.15cm}: ~there are \underline{dense} and \underline{nondense} intervals behind $w$}] at (3,-2.7) {$\cdots$};
 
\end{tikzpicture}
    \end{center}
      \caption{An example with $0,2 \notin P$ and $1 \in P$.}
      \label{fig:appthm2}  
  \end{figure}
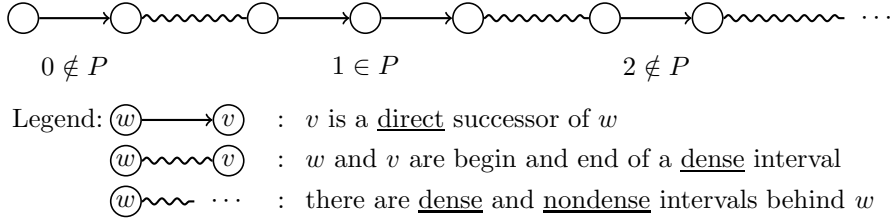 
  In our fragment, it is possible to distinguish between dense and finite intervals.
  We now show how to achieve this.
  In order to encode the alternating sequence of finite and dense intervals that represents a subset $P\subseteq\Nat$,
  we use the free state variable $a$ to mark a state in a dense interval that is directly followed by the first marker. 
  We furthermore use
  the following macros,
  where $x$ and $y$ are state variables that are already bound before the use of the macro, and $r,s,t,u$ are fresh state variables.
  \begin{itemize}
    \item
      \emph{The state named $y$ is a \emph{direct} successor of the state named $x$.}
      It suffices to say that all successors of $x$ are equal to, or occur after, $y$.
      {
      
      \centering
      
        $\directSuc(x,y) := \at_x \Box \dna z.(\at_y z \vee \at_y \Diamond z)$
      
      }
    \item
      \emph{The state named $x$ has no direct predecessor.}
      It suffices to say that, for all states $r$ equal to, or after, the left bound $a$:
      if $r$ is before $x$, then there is a state between $r$ and $x$.
      We work around the implication by saying that one of the following three cases
      occurs: $r$ is after $x$,
      or $r$ equals $x$, or $r$ is before $x$ with a state in between.
      {
      
      \centering
        
        $\noDirectPred(x) := \at_a\Box\dna r.(\at_x \Diamond r \vee \at_x r \vee \at_r\Diamond\Diamond x)$
      
      }
    \item
      \emph{The state named $x$ has a direct predecessor.}
      It suffices to say that there is a state $r$ after $a$ of which $x$ is a direct successor.
      {
      
      \centering
      
        $\directPred(x) := \at_a\Diamond\dna r.\directSuc(r,x)$
      
      }
    \item
      \emph{The interval between states $x,y$ is dense.}
      We say that, for all $r$ with $x < r$ : $r$ is after $y$,
      or $r$ has no direct predecessor.
      {
      
      \centering
      
        $\macro{dense}(x,y) := \at_x\Box\dna r.(\at_y\Diamond r \vee \noDirectPred(r))$
      
      }
    \item
      \emph{The state $x$ is in a separator.}
      This macro says that, for some successor $r$ of $x$,
      the interval between $x$ and $r$ is dense.
      {
      
      \centering
      
        $\macro{sep}(x) := \at_x\Diamond\dna r.\macro{dense}(x,r)$
      
      }
    \item
      \emph{The state $x$ is the begin of a negative marker.}
      This macro says that $x$ has a direct successor that is the begin of a separator,
      and $x$ has no direct predecessor. 
      The latter is necessary to avoid that, in the above example,
      the middle state of a positive marker is mistaken for the begin of a negative marker.
      {
      
      \centering
      
        $\negative(x) := \at_x\Diamond\dna r.(\macro{dirSuc}(x,r) \wedge \macro{sep}(r)) \wedge \noDirectPred(x)$
      
      }
    \item
      \emph{The state $x$ is the begin of a positive marker.}
      Similarly to the above macro, we express that $x$ has
      a direct-successor sequence $r,s$ with $s$ being the begin of a separator,
      and $x$ has no direct predecessor.
      {
      
      \centering
      
        $\positive(x) := \at_x\Diamond\dna r.(\macro{dirSuc}(x,r) \wedge 
                     \Diamond\dna s.(\macro{dirSuc}(r,s) \wedge \macro{sep}(s))) \wedge \noDirectPred(x)$
                     
      }
    \item
      \emph{The state $x$ is in a separator whose end is a marker.}
      This macro says that, for some successor $r$ of $x$,
      the interval between $x$ and $r$ is dense and $r$ is the begin of a marker.
      {
      
      \centering
      
        $\macro{sepM}(x) := \at_x\Diamond\dna r.(\macro{dense}(x,r) \wedge (\negative(r) \vee \positive(r)))$
      
      }
  \end{itemize}
  We now need the following two conjuncts to express that the part of the model starting at $a$ represents a sequence of 
  infinitely many markers.
  \begin{itemize}
    \item
      \emph{$a$ is in a separator that ends with a marker.}~~~
        $\psi_1 := \macro{sepM}(a)$

    \item
      \emph{Every marker has a direct successor marker.}
      We say that every state $r$ after $a$ satisfies one of the following conditions.
      \begin{itemize}
        \item
          $r$ is in a separator---this also includes that $r$ is the end of a marker---that is followed by a marker.
        \item
          $r$ is the begin of a negative marker and its direct successor is the begin of a separator whose end is a marker.
        \item
          $r$ is the begin of a positive marker and its direct 2-step successor is the begin of a separator whose end is a marker.
        \item
          $r$ in the middle of a positive marker, i.e., $r$ has a direct predecessor which is the begin of a positive marker, 
          and $r$'s direct successor is in a separator whose end is a marker.
      \end{itemize}
      \begin{align*}
        \psi_2 :=\, & \at_a \Box\dna r.\Big(\macro{sepM}(r) \\[2px]
                 & \vee \Big(\negative(r) \wedge \Diamond \dna s.(\directSuc(r,s) \wedge \macro{sepM}(s))\Big) \\[2px]
                 & \vee \Big(\positive(r) \wedge \Diamond \dna s.(\directSuc(r,s) \wedge \Diamond\dna t.(\directSuc(s,t) \wedge \macro{sepM}(t)))\Big) \\[2px]
                 & \vee \Big( (\at_a\Diamond\dna s.\directSuc(s,r) \wedge \positive(s)) \wedge \Diamond \dna t.(\directSuc(r,t) \wedge \macro{sepM}(t))\Big)
      \end{align*}
  \end{itemize}
  Finally, we encode formulae $\varphi$ from $\FOL(<,P)$. 
  We assume w.l.o.g.\ that such formulae have the shape 
  $\varphi := Q_1x_1\dots Q_nx_n.\beta(x_1,\dots,x_n)$,
  where $Q_i \in \{\exists,\forall\}$ and $\beta$ is quantifier-free with atoms
  $P(x)$ and $x<y$ for variables $x,y$, such that negations appear only directly before atoms.
  The transformation of $\varphi$ reuses the $x_i$ as state variables
  and proceeds inductively as follows.
  \begin{align*}
    f(P(x_i))                   &~:=~ \positive(x_i)                        \\[4px]
    f(\neg P(x_i))              &~:=~ \negative(x_i)                       \\[4px]
    f(x_i < x_j)                &~:=~ \at_{x_i}\Diamond x_j                   \\[4px]
    f(\neg(x_i < x_j))          &~:=~ \at_{x_i}x_j \vee \at_{x_j}\Diamond x_i \\[4px]
    f(\alpha \wedge \beta)      &~:=~ f(\alpha) \wedge f(\beta)               \\[4px]
    f(\alpha \vee \beta)        &~:=~ f(\alpha) \vee f(\beta)                 \\[2px]
    f(\exists x_i.\alpha)       &~:=~ \at_a\Diamond\dna x_i.\Big((\negative(x_i) \vee \positive(x_i)) \wedge f(\alpha)\Big) \\
    f(\forall x_i.\alpha)       &~:=~ \at_a\Box\dna x_i.\Big(\macro{sep}(x_i) \vee \directPred(x_i) \vee f(\alpha)\Big) \\
  \end{align*}
  The transformation of $\varphi$ into $\MHL(\Diamond,\Box,\dna,\at)$ is now achieved by
  the function $g$ defined as follows.
  \[
    g(\varphi) := \psi_1 \wedge \psi_2 \wedge f(\varphi)
  \]
  It is clear that the reduction function $g$ can be computed in polynomial time.
  The correctness of the reduction is expressed by the following claim. 
  
  \begin{claim}
  For every formula $\varphi$ from $\FOL(<,P)$ holds:
  {

    \centering

      $\varphi\in\SaT[\Nat]_{\FOL}(<,P)$ if and only if $g(\varphi) \in \MSaT[\lin](\Diamond,\Box,\dna,\at)$.

  }
  \end{claim}

The proof of the claim should be clear.
Since $\SaT[\Nat]_{\FOL}(<,P)$ is \nonelementary \cite{stmephdthesis},
it follows that $\MSaT[\lin](\Diamond,\Box,\dna,\at)$ is \nonelementary, too.

Finally, we note that our reduction uses a single free state variable $a$, which could as well be bound to the first state of evaluation.
\end{proof}

The high complexity of $\MSaT[\lin](\Diamond,\Box,\dna,\at)$ relies on the
  possibility that the linear frame alternatingly has dense and non-dense parts. 
If we have the natural numbers as frame for a hybrid language, we lose this possibility.
As a consequence,
the satisfiability problem for monotone hybrid logics over the natural numbers
has a lower complexity than that over linear frames.

\begin{theorem}\label{theo:N-MSAT-PSPACEc}
  $\MSaT[\N](\Diamond,\Box,\dna,\at)$ is \PSPACE-complete.
\end{theorem}

\begin{proof*}{Proof\Reportout{ sketch}.}
Let $\QBFSAT$ be the problem to decide whether a given quantified Boolean formula is valid. 
We show \PSPACE-hardness by a polynomial-time reduction from the $\PSPACE$-complete $\QBFSAT$ to $\MSaT[\N](\Diamond,\Box,\dna,\at)$.
Let $\varphi$ be an instance of $\QBFSAT$ and 
assume w.l.o.g.\ that negations occur only directly in front of atomic propositions.
We define the transformation as $f\colon \varphi \mapsto \dna r. \Diamond \dna s. \Diamond h(\varphi)$
where $h$ is given as follows: let $\psi,\chi$ be quantified Boolean formulae and let $x_k$ be a variable in $\varphi$, then 
  \[
  \renewcommand{\arraystretch}{1.25}
  \begin{array}{@{}l@{\hspace*{1cm}}l@{}}
    h(\exists x_k \psi):=\at_r \Diamond \dna x_k. h(\psi), &
    h(\forall x_k \psi):=\at_r \Box \dna x_k. h(\psi), \\[2px]
    h(\psi \wedge \chi):=h(\psi) \wedge h(\chi), &
    h(\psi \vee \chi):=h(\psi) \vee h(\chi), \\[2px]
    h(\neg x_k):=\at_s \Diamond x_k, &
    h(x_k):=\at_s x_k. \\
  \end{array}
  \]
For example, the QBF $\psi = \forall x\exists y(x \wedge y) \vee (\neg x \wedge \neg y)$ is mapped to 
{

\centering

$f(\varphi) = \dna r. \Diamond \dna s. \Diamond \at_r\Box\dna x_0.\at_r\Diamond\dna x_1.
(\at_s x_0 \wedge \at_s x_1) \vee (\at_s \Diamond x_0 \wedge \at_s \Diamond x_1)$.

}
Intuitively, this construction requires the existence of an initial state named $r$,
a successor state $s$ that represents the truth value $\top$, 
and one or more successor states of $s$ which together represent $\bot$.
The quantifiers $\exists,\forall$ are replaced by the modal operators $\Diamond,\Box$ which range over
$s$ and its successor states. Finally, positive literals are enforced to be true at $s$,
negative literals strictly after $s$.

For every model of $f(\varphi)$, it holds that $r$ is situated at the first state of the model and that state has a successor labelled by $s$. 
By virtue of the function $h$, positive literals have to be mapped to $s$, 
whereas negative literals have to be mapped to some state other than $s$. 
An easy induction on the structure of formulae shows that 
$\varphi \in \QBFSAT$ iff $f(\varphi) \in \MSaT[\N](\Diamond,\Box,\dna,\at)$.


\newcommand{\KN}{K_{\Nat}}
\newcommand{\etaT}{\eta_{\top}}
\newcommand{\Nlt}{(\Nat,<)}
%
We obtain \PSPACE-membership via a polynomial-time reduction 
from $\MSaT[\N](\Diamond,\Box,\dna,\at)$
to the satisfiability problem $\SaT[\Nat]_{\FOL}(<)$
for the fragment of first-order logic with the relation ``$<$'' interpreted over the natural numbers.
Let the first order language contain all members of \SVAR\ as variables
and all members of \NOM\ as constants.
Based on the standard translation from hybrid to first-order logic \cite{cafr05},
we devise a reduction $H$ that maps
hybrid formulae $\varphi$ and variables or constants $z$
to first-order formulae.
  \[
  \renewcommand{\arraystretch}{1.25}
  \begin{array}{@{}l@{\hspace*{.5cm}}l@{}}
    H(p,z) := \top \text{ for $p\in\PROP$}                 
          & H(v,z) :=  v=z ~~\text{ for $v\in\SVAR\cup\NOM$} \\[2px]
    H(\alpha\wedge\beta,z) := H(\alpha,z) \wedge H(\beta,z) 
          &  H(\alpha\vee\beta,z) := H(\alpha,z) \vee H(\beta,z) \\[2px]
    H(\Diamond \alpha,z) := \exists t (z<t \wedge H(\alpha,t)) 
          &  H(\Box \alpha,z) := \forall t (z<t \rightarrow H(\alpha,t)) \\[2px]
    H(\dna x.\alpha,z) := \exists x (x=z \wedge H(\alpha,z))
					 & H(\at_x \alpha,z) := H(\alpha,x)
  \end{array}
  \]
In the $\Diamond$, $\Box$ and $\at$-cases we deviate from the usual definition of the standard translation
because we do not insist on using only two variables in addition to \SVAR---therefore it suffices to require that $t$ is a fresh variable---%
and we allow constants in the second argument.

For a first-order formula $\psi$ with variables in $\SVAR$ and an assignment $g:\SVAR\rightarrow\Nat$,
let $\psi[g]$ denote the first-order formula that is obtained from $\psi$
by substituting every free occurrence of $x\in\SVAR$ by the first-order term that describes $g(x)$.

\begin{claim}
For every instance $\varphi$ of $\MSaT[\N](\Diamond,\Box,\dna,\at)$,
every assignment $g:\SVAR\rightarrow\Nat$ and every $n\in\Nat$, it holds that: ~~
$g,n \models \varphi$ if and only if $\Nlt \models H(\varphi,z)[g^{z}_n]$,
where $z$ is a new variable that does not occur in $\varphi$.
\end{claim}

\Reportin{
\begin{proofofclaim}
We prove the claim inductively on the construction of $\varphi$.
\begin{description}
\item[$\varphi=v$ for $v\in\SVAR$:]

\begin{tabular}[t]{lcl}
$g,n \models v$ & iff$_{(1)}$ & $g(v)=n$ \\ 
  & iff$_{(2)}$ & $g^{z}_{n}(v)=g^{z}_{n}(z)$ \\
  & iff$_{(3)}$ & $\Nlt \models (v=z)[g^{z}_n]$.
\end{tabular}

\noindent
Justifications for the equivalences:
(1) is by the definition of $\models$ for hybrid logic,
(2) extends $g$ by the new variable $z$, and (3)  uses the definition of $\models$ for first-order logic over $\Nlt$.

\item[$\varphi=\alpha\wedge\beta$ resp. $\varphi=\alpha\vee\beta$:] 
     straightforward.

\item[$\varphi=\Diamond \alpha$:]

\begin{tabular}[t]{lcl}
$g,n\models \Diamond \alpha$ & iff$_{(1)}$ & 
  $\exists t'>n: g,t' \models \alpha$ \\
  & iff$_{(2)}$ & $\exists t'>n: \Nlt \models H(\alpha,t)[g^{t}_{t'}]$ \\
  & iff$_{(3)}$ & $\Nlt \models \exists t (z<t \wedge H(\alpha, t))[g^{z}_n]$.
\end{tabular}

\noindent
(1) and (2) are by definition resp. by induction hypothesis.
For (3), notice that the variable $t$ may appear free in $H(\alpha,t)$
but it does not appear free in  $\exists t (z<t \wedge H(\alpha, t))$.
The equivalence then follows by the semantics of the considered first-order logic.

\item[$\varphi=\Box \alpha$:]

\begin{tabular}[t]{lcl}
$g,n\models \Box \alpha$ & iff$_{(1)}$ & 
    $\forall t'>n: g,t' \models \alpha$ \\
    & iff$_{(2)}$ & $\forall t'>n: \Nlt \models H(\alpha,t)[g^{t}_{t'}]$ \\
    & iff$_{(3)}$ & $\Nlt \models \forall t (z<t \rightarrow H(\alpha, t))[g^{z}_n]$.
\end{tabular}

\noindent
(1) and (2) are by definition resp. by induction hypothesis.
The arguments for (3) are as in the case above.

\item[$\varphi=\dna x. \alpha$:]

\begin{tabular}[t]{lcl}
$g,n\models \dna x. \alpha$ & iff$_{(1)}$ &
    $g^x_n ,n\models \alpha$ \\
    & iff$_{(2)}$ & $\Nlt \models H(\alpha,z)[(g^x_n)^z_n]$ \\
    & iff$_{(3)}$ & $\Nlt \models \exists x (x=z \wedge H(\alpha,z))[g^z_n]$.
\end{tabular}

\noindent
(1) and (2) are from the definition of $\dna$ and from the induction hypothesis.
Eventually, (3) follows from the semantics of FOL over $\Nlt$.

\item[$\varphi=\at_x \, \alpha$:]

\begin{tabular}[t]{lcl}
$g,n\models \at_x \, \alpha$ & iff$_{(1)}$ &
        $g,g(x)\models \alpha$ \\
        & iff$_{(2)}$ & $\Nlt \models H(\alpha,z)[g^z_{g(x)}]$ \\
        & iff$_{(3)}$ & $\Nlt \models \exists z (x=z \wedge H(\alpha,z))[g]$ \\
        & iff$_{(4)}$ & $\Nlt \models \exists z (x=z \wedge H(\alpha,z))[g^z_n]$.
\end{tabular}

\noindent
(1) and (2) are from the definition of $\dna$ and from the induction hypothesis.
Now, (3) follows from the semantics of FOL over $\Nlt$.
Notice that $z$ does not appear free in $\exists z (x=z \wedge H(\alpha,z))$.
This proves Equivalence (4).
\end{description}
This concludes the proof of the claim.
\end{proofofclaim}
}

Now, $\varphi\in \MSaT[\N](\Diamond,\Box,\dna,\at)$ if and only if $g,0\models \varphi \vee \Diamond \varphi$
for some assignment $g$.
By the above claim,
this is equivalent to $\Nlt \models H(\varphi \vee \Diamond\varphi,z)[g^{z}_0]$ for some $g$ and a new variable $z$,
which can also be expressed as $\Nlt \models \forall x (\neg(x<z) \wedge H(\varphi \vee \Diamond\varphi,z))$.
This shows that $\MSaT[\N](\Diamond,\Box,\dna,\at)$ is polynomial-time reducible
to $\SaT[\Nat]_{\FOL}(<)$, which was shown to be in \PSPACE in \cite{fera79}.
Therefore, $\MSaT[\N](\Diamond,\Box,\dna,\at)$ is in \PSPACE.
\end{proof*}


\section{The easy cases: \texorpdfstring{$\NC1$ and $\L$ results}{NC1 and L results}}
\noindent
In this section, we show that the fragments without the $\Diamond$-operator
have an easy satisfiability problem.
Our results can be structured into four groups.
First, we consider fragments without modal operators.
For these fragments we obtain $\NC1$-completeness.
Simply said, without negation and $\Diamond$ we cannot express that
two nominals or state variables are not bound to the same state.
Therefore, the model that binds all variables to the first state satisfies 
every satisfiable formula in this fragment.

\begin{lemma}\label{lemma:oneworld}
Let $F_0=(\{0\},\emptyset)$ and $g_0(y)=0$ for every $y \in \SVAR$.
Then $\varphi \in \MSaT[\lin](\dna, \at)$ (resp. $\varphi \in \MSaT[\Nat](\dna, \at)$)
if and only if $F_0,g_0,0\models\varphi$.
\end{lemma}

\begin{proof}
The implication direction from left to right follows from the monotonicity of the considered formulas.
For the other direction, notice that $F_0\in\lin$.
For frame class \Nat,
note that if $F_0,g_0,0\models\varphi$ and $\varphi$ has no modal operators, then
$g_0,0\models\varphi$.
\end{proof}

\begin{theorem} \label{thm:all-M-NC1-compl}
Let $O \subseteq \{\dna, \at \}$. 
Then $\MSaT[\lin](O)$ and $\MSaT[\Nat](O)$ are $\NC1$-complete.
\end{theorem}

\begin{proof}
$\NC1$-hardness of $\MSaT[\Fclass{F}](\emptyset)$ follows immediately
from the $\NC1$-com\-plete\-ness of the Formula Value Problem for propositional formulae \cite{bus87}.
It remains to show that $\MSaT[\lin](\dna, \at)$ and $\MSaT[\Nat](\dna, \at)$ are in $\NC1$.
In order to decide whether $\varphi$ is in $\MSaT[\lin](\dna, \at)$,
according to Lemma~\ref{lemma:oneworld} it suffices to check
whether the propositional formula obtained from $\varphi$
deleting all occurrences of $\dna x.$ and $\at_x$, 
is satisfied by the assignment that sets all atoms to \emph{true}.
According to \cite{bus87} this can be done in $\NC1$.
Since $\MSaT[\lin](\dna, \at)=\MSaT[\Nat](\dna, \at)$ by Lemma~\ref{lemma:oneworld},
we obtain the same for $\MSaT[\Nat](\dna, \at)$.
\end{proof}

\par\medskip\noindent
Second, we consider fragments with the $\Box$-operator over linear frames.
We can show $\NC1$-completeness here, too.
The main reason is that (sub-)formulas that begin with a $\Box$
are satisfied in a state that has no successor.
Therefore similar as above, every formula of this fragment that is
satisfiable over linear frames is satisfied by a model with only one state.

\begin{theorem} \label{lem:lin-M(B!@)-NC1-ub}
$\MSaT[\lin](\Box, \dna, \at)$ is $\NC1$-complete.
\end{theorem}

\begin{proof}
$\NC1$-hardness follows from Theorem~\ref{thm:all-M-NC1-compl}.
It remains to show that $\MSaT[\lin](\Box, \dna, \at)\in\NC1$.
We show that essentially the $\Box$-operators can be ignored.
\begin{claim}
  \label{claim:correctness_reduction_linMSAT(B!@)}
  $\MSaT[\lin](\Box, \dna, \at) \leqcd \MSaT[\lin](\dna, \at)$.
\end{claim}

\begin{proofofclaim}
For an instance $\varphi$ of $\MSaT[\lin](\Box, \dna, \at)$,
let $\varphi''$ be the formula obtained from $\varphi$  
by replacing every subformula $\Box\psi$ of $\varphi$ with the constant $\top$.
Then $\varphi''$ is an instance of  $\MSaT[\lin](\dna, \at)$.
If $\varphi \in \MSaT[\lin](\Box, \dna, \at)$, then $\varphi''\in \MSaT[\lin](\dna, \at)$ due to the monotonicity of $\varphi$.
On the other hand, 
if $\varphi''\in \MSaT[\lin](\dna, \at)$, then $K_0,g,0\models\varphi''$ (Lemma~\ref{lemma:oneworld}).
Since $K_0,g,0\models\Box\alpha$ for every $\alpha$,
we obtain $K_0,g,0\models\varphi$, hence $\varphi \in \MSaT[\lin](\Box, \dna, \at)$.
As such simple substitutions can be realized using an $\AC0$-circuit,
the stated reduction is indeed a valid $\leqcd$-reduction from $\MSaT[\lin](\Box,\dna,\at)$ to $\MSaT[\lin](\dna,\at)$.
\end{proofofclaim}

Since $\MSaT[\lin](\dna, \at)\in\NC1$ (Theorem~\ref{thm:all-M-NC1-compl}) and $\NC1$ is closed downwards under $\leqcd$,
it follows from the Claim that  $\MSaT[\lin](\Box, \dna, \at)\in\NC1$.
\end{proof}

It is clear that this argument does not apply to the natural numbers.

\par\medskip\noindent
Third, we show $\NC1$-completeness for the fragments with $\Box$
and one of $\dna$ and $\at$ over $\Nat$.
They receive separate treatment
because, in \Natless, every state has a successor, and
therefore $\Box$-subformulas cannot be satisfied as easily as above.
It turns out that the complexity of the satisfiability problem increases only
if both hybrid operators can be used.

\begin{theorem}\label{thm:NC1-compl-BA}
$\MSaT[\N](\Box,\at)$ is $\NC1$-complete.
\end{theorem}

\begin{proof*}{Proof sketch.}
$\NC1$-hardness follows from Theorem~\ref{thm:all-M-NC1-compl}.

For the upper bound,
we distinguish occurrences of nominals that are either \emph{free}, 
or that are \emph{bound} by a $\Box$, or that are bound by an $\at$.
Simply said, a free occurrence of $i$ in $\alpha$
is bound by $\Box$ in $\Box\alpha$ and bound by $\at$ in $\at_x \alpha$ (even if $x\not=i$).
Since the assignment $g$ is not relevant for the considered fragment,
we write $K,w\models\alpha$ for short instead of $K,g,w\models\alpha$.

\begin{claim}
Let $\alpha'$ be the formula obtained from $\alpha$
by replacing every occurrence of a nominal that is bound by $\Box$ with $\bot$,
and let $\eta$ be a valuation.
If $\eta,k\models\alpha$, then $\eta,k\models\alpha'$.
\end{claim}

Moreover, it turns out that binding every nominal to the initial state
suffices to obtain a satisfying model.

\begin{claim}
$\varphi\in \MSaT[\N](\Box,\at)$ if and only if $\eta_0,0\models\varphi$
with $\eta_0(x)=\{0\}$ for every $x \in \NOM$.
\end{claim}

Both claims together yield that, 
in order to decide $\varphi\in \MSaT[\N](\Box,\at)$,
it suffices to check whether $\eta_0,0\models\varphi'$.
No nominal in $\varphi'$ occurs bound by a $\Box$-operator.
Therefore for every subformula $\Box\alpha$ of $\varphi'$  and for every $k$ holds:
$\eta_0,k\models\alpha$ if and only if $\eta_0,0\models\alpha$.
All nominals that occur free or bound by an $\at$ 
evaluate to \emph{true} in state $0$ via $\eta_0$.
Therefore, in order to decide $\eta_0,0\models\varphi'$,
it suffices to ignore all $\Box$ and $\at$-operators of $\varphi'$
and evaluate it as a propositional formula under assignment $\eta_0$
that sets all atoms of $\varphi'$ to \emph{true}.
This can be done in $\NC1$ \cite{bus87}.
The complete proof can be found in \Reportout{the Technical Report version of this paper~\cite{GMMSTW11}}\Reportin{Appendix \ref{app:NC1_proof_BA}}.
\end{proof*}

Next, we consider $\MSaT[\N](\Box,\dna)$.
According to our remarks in Section \ref{sec:prelims} about notational convenience,
we assume that there are no nominals in $\MHL(\Box,\dna)$.

\begin{theorem}\label{thm:NC1-compl-BD}
$\MSaT[\N](\Box,\dna)$ is $\NC1$-complete.
\end{theorem}

\begin{proof*}{Proof sketch.}
Now, we distinguish occurrences of state variables as the occurrences in the proof sketch above.
They are either \emph{free}, or they are \emph{bound} by a $\Box$, or they are \emph{bound} by $\dna$.
Note that this phrasing differs from the standard usage of the terms `free' and `bound' in the context of state variables.
A free occurrence of $i$ in $\alpha$ is bound by $\Box$ in $\Box\alpha$, as above.
It is bound by $\dna$ in $\dna i . \alpha$ only.
Notice that $y$ occurs free in $\dna x . y$ (for $x\not=y$).

\begin{claim}
Let $\alpha'$ be the formula obtained from $\alpha$
by replacing every occurrence of a state variable that is bound by $\Box$ with $\bot$,
and let $g$ be an assignment.
If $g,k\models\alpha$, then $g,k\models\alpha'$.
\end{claim}

\begin{claim}
$\varphi\in \MSaT[\N](\Box,\dna)$ if and only if $g_0,0\models\varphi$, 
for $g_0(x)=0$ for every $x \in \SVAR$.
\end{claim}

Both claims together yield that,
in order to decide $\varphi\in \MSaT[\N](\Box,\dna)$,
it suffices to check whether $g_0,0\models\varphi'$.
No state variable in $\varphi'$ occurs bound by a $\Box$-operator.
Therefore for every subformula $\Box\alpha$ of $\varphi'$  and for every $k$ holds:
$g_0,k\models\alpha$ if and only if $g_0,0\models\alpha$.
All occurrences of state variables in $\varphi'$ that are bound by $\dna$
evaluate to \emph{true}, because no $\Box$ occurs ``between'' the 
binding $\dna i$ and the occurrence of $i$,
which means that the state where the variable is bound is the same as where the variable is used.
All free occurrences of state variables  
evaluate to \emph{true} in state $0$ due to $g_0$.
Therefore, in order to decide $g_0,0\models\varphi'$,
it suffices to ignore all $\Box$ and $\dna$-operators of $\varphi'$
and evaluate it as a propositional formula under an assignment
that sets all atoms to \emph{true}.
This can be done in $\NC1$ \cite{bus87}.
The complete proof can be found in \Reportout{the Technical Report version of this paper~\cite{GMMSTW11}}\Reportin{Appendix \ref{app:NC1_proof_BD}}.
\end{proof*}

The fourth part deals with the fragment with $\Box$ and both $\dna$ and $\at$ over the natural numbers.

\begin{lemma} \label{lem:N-M(B!@)-L-lb}
  $\MSaT[\Nat](\Box, \dna, \at)$ is $\L$-hard.
\end{lemma}

\begin{proof}
This proof is very similar to the proof of Theorem~3.3.\ in \cite{MMSTWW09}. 
We give a reduction from the problem \emph{Order between Vertices} (\ORD)
which is known to be $\L$-complete \cite{etessami:1997} and defined as follows.

 \problemdef%
 	{$\ORD$}
 	{A finite set of vertices $V$, a successor-relation $S$ on $V$, and two vertices $s,t\in V$.}
 	{Is $s\leqslant_S t$, where $\leqslant_S$ denotes the unique total order induced by $S$ on $V$?}
	
Notice that $(V,S)$ is a directed line-graph. Let $(V,S,s,t)$ be an instance of $\ORD$. 
We construct an $\MHL(\Box,\dna,\at)$-formula $\varphi$ that is satisfiable if and only if $s\leqslant_S t$. We use $V=\{v_0,v_1,\ldots,v_n\}$ as state variables. 
The formula $\varphi$ consists of three parts. 
The first part binds all variables except $s$ to one state and the variable $s$ to a successor of this state. 
The second part of $\varphi$ binds a state variable $v_l$ to the state labeled by $s$ iff $s\leqslant_S v_l$. 
Let $\alpha$ denote the concatenation of all $\at_{v_{k}}\dna v_{l}$ with $(v_k,v_l)\in S$ and $v_l\not=s$, and $\alpha^n$ denotes the $n$-fold concatenation of $\alpha$. 
Essentially, $\alpha^n$ uses the assignment to collect eventually all $v_i$ with $s\leqslant_S v_i$ in the state labeled $s$. 
The last part of $\varphi$ checks whether $s$ and $t$ are bound to the same state after this procedure. 
That is,
$
	\varphi = \dna v_0.\dna v_1.\dna v_2. \cdots \dna v_n. \Box \dna s. ~ \alpha^{n} ~ \at_s t.
$
To prove the correctness of our reduction, we show that $\varphi$ is satisfiable if and only if $s\leqslant_S t$.

Assume $s\leqslant_S t$.
For an arbitrary assignment $g$, one can show inductively that $g,0\models \dna v_0.\dna v_1. \cdots \dna v_n. \Box \dna s. ~ \alpha^i ~ \at_s r$
for $i=0,1,\ldots,n$ and for all $r$ that have distance $i$ from $s$.
Therefore it eventually holds that $g,0\models \varphi$.
For $s\not\leqslant_S t$ we show that $g,n\not\models\varphi$ for any assignment $g$ and natural number $n$.
Let $g_0$ be the assignment obtained from $g$ after the bindings in the prefix $\dna v_0.\dna v_1. \cdots \dna v_n. \Box \dna s$ of $\varphi$, 
and let $g_i$ be the assignment obtained from $g_0$ after evaluating the prefix of $\varphi$ up to and including $\alpha^i$. 
It holds that $g_i(s)\not=g_i(t)=0$ for all $i=0,1,\ldots,n$. 
This leads to $g_{n},0\not\models \at_s t$ and therefore $g,0\not\models\varphi$.
\end{proof}

For the upper bound, we establish a characterisation of the satisfaction relation
that assigns a \emph{unique} assignment and state of evaluation to every subformula
of a given formula $\varphi$. Using this new characterisation, we devise a
decision procedure that runs in logarithmic space and consists of two steps:
it replaces every occurrence of any state variable $x$ in $\varphi$ with
1 if its state of evaluation agrees with that of its $\dna x$-superformula,
and with 0 otherwise; it then removes all $\Box$-, $\dna$- and $\at$-operators
from the formula and tests whether the resulting Boolean formula is valid.


\begin{theorem}
  \label{thm:N-Bda-in-L}
  $\MSaT[\Nat](\Box,\dna,\at)$ is in \LOGSPACE.
\end{theorem}

\Reportout{The proof is technically involved and can be found in the Technical Report version of this paper~\cite{GMMSTW11}.}
\Reportin{The proof can be found in Appendix \ref{app:N-Bda-in-L}.}

\section{The intermediate cases: \texorpdfstring{$\NP$ results}{NP results}}
\noindent
After we have seen that all fragments without $\Diamond$ have an easy satisfiability problem,
we show that $\Diamond$ together with the use of nominals 
makes the satisfiability problem $\NP$-hard.
Recall that, owing to the presence of nominals, $\MHL(\Diamond)$ is not just modal logic with the $\Diamond$-operator.
The absence of $\dna$ makes assignments superfluous: we write $K,w \models \varphi$ instead of $K,g,w \models \varphi$.

\begin{lemma} \label{lem:M-NP-lb}
$\MSaT[\lin](\Diamond)$ and $\MSaT[\Nat](\Diamond)$ both are $\NP$-hard.
\end{lemma}

\begin{proof*}{Proof\Reportout{ sketch}.}
We reduce from $\dSAT$. 
Let $\varphi=c_1 \wedge \ldots \wedge c_n$ be an instance of $\dSAT$ 
with clauses $c_1, \dots, c_n$ (where $c_i=(l_1^i\vee l_2^i\vee l_3^i)$ for literals $l^i_j$) 
and variables $x_1, \dots, x_m$.
We define the transformation as 
  \[
    f\colon \varphi \mapsto 
     \Diamond(i_0 \wedge \Diamond i_1) ~\wedge~
     \Bigg(\bigwedge_{\ell=1}^m \Diamond (i_0 \wedge x_{\ell}) \vee \Diamond (i_1 \wedge x_{\ell})\Bigg) \wedge 
     h(\varphi),
  \] 
  where $i_0,i_1$ and all $x_{\ell}$ are nominals, and
  the function $h$ is defined as follows: let $l_k^j$ be a literal in clause $c_j$, then 
  \begin{align*}
    h(l_k^j)& :=  \begin{cases}
                      (i_1 \wedge x), \text{ if } l_k^j=x \\
                      (i_0 \wedge x), \text{ if } l_k^j=\neg x
                    \end{cases} \\[2px]
    h(c_j)  & := \Diamond( h(l_1^j) \vee h(l_2^j) \vee h(l_3^j)),\quad \text{where } c_j = (l_1^j \vee l_2^j \vee l_3^j); \\[4px]
    h(c_1 \wedge \dots \wedge c_n) 
            & := h(c_1) \wedge \dots \wedge h(c_n).
  \end{align*}

Notice that $f$ turns variables in the $\dSAT$ instance into \emph{nominals} in the $\MSaT[\lin](\Diamond)$ instance.
The part  $\Diamond(i_0 \wedge \Diamond i_1)$ enforces the existence of two successors $w_1$ and $w_2$ of the state satisfying $f(\varphi)$. 
The part $\bigwedge_{\ell=1}^m \Diamond (i_0 \wedge x_{\ell}) \vee \Diamond (i_1 \wedge x_{\ell})$  simulates the assignment of the variables in $\varphi$,
enforcing that each $x_{\ell}$ is true in either $w_1$ or $w_2$. 
The part $h(\varphi)$ then simulates the evaluation of $\varphi$ on
the assignment determined by the previous parts.
With the following claim
 $\NP$-hardness of $\MSaT[\lin](\Diamond)$ follows.

\begin{claim}
  \label{claim:correctness_red_3SAT_lin-MSAT(D)}
  $\varphi \in \dSAT$ if and only if $h(\varphi) \in \MSaT[\lin](\Diamond)$.
\end{claim}

\Reportin{\begin{proofofclaim} 
We first show that $h(\varphi) \in \MSaT[\lin](\Diamond)$ implies $\varphi \in \dSAT$.
If $K,w_0 \models h(\varphi)$ with $K=(W,<,\eta)$, then the following holds.
Let $w_1=\eta(i_0)$, $w_2=\eta(i_1)$, and
\begin{itemize}
        \itemsep0pt
        \parsep0pt
        \item $\{w_0,w_1,w_2\} \subseteq W$ with $w_0,w_1,w_2$ pairwise different;
        \item $w_0<w_1<w_2$;
        \item for all $x_j$ with $1 \leqslant j \leqslant m$ : $\eta(x_j) \subseteq \{w_1,w_2\}$.
\end{itemize}
We build a propositional logic assignment $\beta = (\beta_1 \dots \beta_m)$ that satisfies $\varphi$,
where $\beta_i \in \{\bot,\top\}$ is the truth value for $x_i$, as follows.
$\beta_j=\bot$ if $g(i_0) = g(x_j)$, and $\beta_j=\top$ if $g(i_1) = g(x_j)$.
From the construction of $h(\varphi)$, it clearly follows that $\beta$ satisfies $\varphi$.
\par
For the converse direction, suppose that $\varphi$ is satisfied by the propositional logic assignment $\beta = (\beta_1 \dots \beta_m)$. 
We construct a linear model $K:=(W,<,\eta)$ containing a state $w$ such that $K,w \models h(\varphi)$.

{\centering
  \parbox{.5\textwidth}{%
    \begin{align*}
      W &:= \{w,w_0,w_1\}                 \\
      < &:~~~ w<w_0<w_1 \\
    \end{align*}%
  }
  \parbox{.4\textwidth}{%
    \begin{align*}
      \eta(i_j) &:= w_j \text{ for } j \in \{0,1\} \\
      \eta(x_j) &:= \begin{cases}
                      w_0, \text{ if } \beta_j=\bot \\
                      w_1, \text{ if } \beta_j=\top
                    \end{cases}
    \end{align*}%
  }
}

It follows from the construction of $K$ that $K,w \models h(\varphi)$. 
The conjunct $h(\varphi)$ is of the form
$
  (h(l_1^1) \vee h(l_1^2) \vee h(l_1^3)) \wedge \dots \wedge (h(l_n^1) \vee h(l_n^2) \vee h(l_n^3)).
$
Hence, under $\beta$, at least one literal in every clause evaluates to true. 
The variable in this literal satisfies the same clause in $h(\varphi)$. 
Hence every clause in $h(\varphi)$ is satisfied in $w$ in $K$. Therefore, $K,w \models h(\varphi)$.
\end{proofofclaim}%
} 

Using this claim,
 $\NP$-hardness of $\MSaT[\lin](\Diamond)$ follows.
It is straightforward to show that $\dSAT$ reduces to $\MSaT[\Nat](\Diamond)$ using the same reduction. 
\end{proof*}

\noindent
We will now establish \NP-membership of the problems 
$\MSaT[\Fclass{F}](\Diamond,\Box,\dna)$,
$\MSaT[\Fclass{F}](\Diamond,\Box,\at)$, and
$\MSaT[\Fclass{F}](\Diamond,\dna,\at)$ for $\Fclass{F} \in \{\lin,\N\}$.
For the first two, this follows from the literature, see Theorem \ref{theo:known_for_and_or_neg} \ref{it:known_for_and_or_neg_over_lin_N}.
For the third, we observe that all modal and hybrid operators in a formula $\varphi$ from the fragment $\MHL(\Diamond,\dna,\at)$ 
are translatable into FOL by the standard translation using no universal quantifiers. 
The existential quantifiers introduced by the binder can be skolemised away,
which corresponds to removing all binding from $\varphi$ and replacing each state variable with a fresh nominal.
The correctness of this translation is proven in \cite{cafr05}.
Hence, $\MSaT[\Fclass{F}](\Diamond,\dna,\at)$ polynomial-time reduces to $\MSaT[\Fclass{F}](\Diamond,\at)$.

\begin{lemma} \label{lem:M(D!@)-NP-ub}
$\MSaT[\lin](\Diamond,\dna,\at)$ and $\MSaT[\N](\Diamond,\dna,\at)$ are in $\NP$.
\end{lemma}

From the lower bounds in Lemma \ref{lem:M-NP-lb}
and the upper bounds  in
Theorem \ref{theo:known_for_and_or_neg} \ref{it:known_for_and_or_neg_over_lin_N}
and Lemma \ref{lem:M(D!@)-NP-ub}, we obtain the following theorem.

\begin{theorem} \label{thm:lin,N-M-NP-compl}
Let 
$\{ \Diamond \} \subseteq O$, and 
$O \subsetneq \{\Diamond, \Box, \dna, \at \}$.
Then $\MSaT[\lin](O)$ and $\MSaT[\N](O)$ are $\NP \emph{-complete}$.
\end{theorem}

In addition to the NP-membership of the fragments captured by Theorem \ref{thm:lin,N-M-NP-compl}, we are interested in their model-theoretic properties.
We show that these logics enjoy a kind of linear-size model property, precisely a quasi-quadratic size model property: over the natural numbers, every satisfiable formula
has a model where two successive nominal states have at most linearly many intermediary states, and the states behind the last
such state are indistinguishable. 
This property allows for an alternative worst-case decision procedure for satisfiability
that consists of guessing a linear representation of a model of the described form and symbolically model-checking the input formula
on that model. Over general linear frames, which may have dense intervals, we formulate the model property in a more general way
and prove it using additional technical machinery to deal with density. However, the result then carries over to the rationals,
where we are not aware of any upper complexity bound in the literature.

In \cite{sicl85}, Sistla and Clarke showed a variation of the linear-size model property for LTL(F), which corresponds to $\HL(\Diamond,\Box)$ over \Nat:
whenever $\varphi \in \HL(\Diamond,\Box)$ is satisfiable over \Nat, then it is satisfiable in the initial state of a model over \Nat
which has a linear-sized prefix \textsf{init} and a remainder \textsf{final} such that \textsf{final} is maximal with respect to the property
that every type (set of all atomic propositions true in a state) occurs infinitely often, and \textsf{final} contains only linearly many types.
Such a structure can be guessed in polynomial time, represented in polynomial space and model-checked in polynomial time.
While it is straightforward to extend Sistla and Clarke's proof to cover nominals and the \at\ operator,
it will not go through if density is allowed (frame class $\lin$).

We establish that $\MHL(\Diamond,\Box,\at)$ over $\lin$ has a quadratic size model property, and we subsequently show how to extend the result to the other fragments
from Theorem \ref{thm:lin,N-M-NP-compl} and how to restrict them to \Nat.

\begin{theorem}
  \label{thm:QLMP_specific}
  $\MHL(\Diamond,\Box,\at)$ has the quasi-quadratic size model property
  with respect to $\lin$ and $\Nat$.
\end{theorem}

The proof can be found in \Reportout{the Technical Report version of this paper~\cite{GMMSTW11}.}\Reportin{Appendix \ref{app:QLMP_specific}.}

As an immediate consequence, the model property in Theorem \ref{thm:QLMP_specific} carries over to the subfragments
$\MHL(\Diamond,\Box)$, $\MHL(\Diamond,\at)$, $\MHL(\Box,\at)$, $\MHL(\Diamond)$, $\MHL(\Box)$, $\MHL(\at)$, and $\MHL(\emptyset)$.
Moreover,
our arguments in the proofs of Theorems~\ref{lem:lin-M(B!@)-NC1-ub} and \ref{thm:N-Bda-in-L}
can be used to transfer it to $\MHL(\Box,\dna,\at)$. 
Together with the observations that 
\begin{itemize}
  \item
    $\MHL(\Diamond,\dna,\at)$ is no more expressive than $\MHL(\Diamond,\at)$ (see the explanation before Lemma \ref{lem:M(D!@)-NP-ub}), and
  \item
    $\MHL(\Diamond,\Box,\dna)$ is no more expressive than $\MHL(\Diamond,\Box)$ (because, without $\at$, one cannot jump to named states),
\end{itemize}
we obtain the following generalisation of Theorem \ref{thm:QLMP_specific}.
%

\begin{corollary}
  \label{cor:QLMP_general}
  Let $O \subsetneq \{\Diamond,\Box,\dna,\at\}$.
  Then $\MHL(O)$ has the quasi-quadratic size model property
  with respect to $\lin$ and $\Nat$.
\end{corollary}


\section{Conclusion}

We have completely classified the complexity of all fragments of hybrid logic with monotone Boolean operators
obtained from arbitrary combinations of four modal and hybrid operators, over linear frames and the natural numbers.
Except for the largest such fragment over linear frames, all fragments are of elementary complexity.
We have classified their complexity into \PSPACE-complete, \NP-complete and tractable
and shown that the tractable cases are complete for either \NC1 or \LOGSPACE.
Surprisingly, while the largest fragment is harder over linear frames than over \Natless,
the largest $\Diamond$-free fragment is easier over linear frames than over \Natless.

The question remains whether the \PSPACE-complete largest fragment over \Natless admits some quasi-polynomial size model property.
Furthermore, this study can be extended in several possible ways: by allowing negation on atomic propositions,
by considering frame classes that consist only of dense frames, such as \Ratless,
or by considering arbitrary sets of Boolean operators in the same spirit as in \cite{MMSTWW09}.
For atomic negation, it follows quite easily that the largest fragment is of \nonelementary complexity over \Natless, too,
and that all fragments except $O = (\Box,\dna,\at)$ are \NP-complete. However, our proof of the quasi-quadratic model property
does not immediately go through in the presence of atomic propositions. 
Over \Ratless, we conjecture that all fragments, except possibly for the largest one, have the same complexity and model properties as over \Natless.


\newpage
\appendix
\section*{Appendix}

\Reportin{
\section{Proof of Theorem \ref{thm:NC1-compl-BA}}
\label{app:NC1_proof_BA}

\textbf{Theorem \ref{thm:NC1-compl-BA}}
\emph{$\MSaT[\N](\Box,\at)$ is $\NC1$-complete.}

\begin{proof}
$\NC1$-hardness follows from Theorem~\ref{thm:all-M-NC1-compl}.

For the upper bound,
we distinguish occurrences of nominals that are either \emph{free}, 
or that are \emph{bound} by a $\Box$, or that are bound by an $\at$.
Simply said, a free occurrence of $i$ in $\alpha$
is bound by $\Box$ in $\Box\alpha$ and bound by $\at$ in $\at_x \alpha$ (even if $x\not=i$).
Since the assignment $g$ is not relevant for the considered fragment,
we write $K,w\models\alpha$ for short instead of $K,g,w\models\alpha$.

\begin{claim}
Let $\alpha'$ be the formula obtained from $\alpha$
by replacing every occurrence of a nominal that is bound by $\Box$ with $\bot$,
and let $\eta$ be a valuation.
If $\eta,k\models\alpha$, then $\eta,k\models\alpha'$.
\end{claim}

\begin{proofofclaim}
We use induction on the construction of $\varphi$.
The base case for $\varphi\in\PROP\cup\NOM$ is straightforward,
as is the inductive step for $\varphi=\alpha\vee\beta$ and $\varphi=\alpha\wedge\beta$,
and even for $\varphi=\at_x\alpha$.
It remains to consider the case $\varphi=\Box\alpha$.
If $\eta,k\models\Box\alpha$, then for all $k'>k$: $\eta,k'\models \alpha$ (by semantics of $\Box$)
and by inductive hypothesis follows for all $k'>k$: $\eta,k'\models \alpha'$.
Assume that in $\Box(\alpha')$ there occurs a nominal $i$ that is bound by the initial $\Box$-operator.
Since for all $k>k'$ holds $\eta,k'\models \alpha'$,
there is some $\ell>\max\bigcup_{j\in\NOM} \eta(j)$ with $\eta,\ell\models \alpha'$.
Therefore $\eta,\ell\models \alpha'[i/\bot]$, and by the monotonicity of $\alpha'$ and the properties of $\eta$
it follows that for all $k'>k$ holds $\eta,k'\models \alpha'[i/\bot]$.
In this way, all nominals bound by the initial $\Box$-operator can be replaced by $\bot$,
and it follows that $\eta,k\models (\Box(\alpha'))'$.
Since $(\Box(\alpha'))'=(\Box\alpha)'$, the claim follows.
\end{proofofclaim}

\begin{claim}
$\varphi\in \MSaT[\N](\Box,\at)$ if and only if $\eta_0,0\models\varphi$
with $\eta_0(x)=\{0\}$ for every $x \in \NOM$.
\end{claim}

\begin{proofofclaim}
We use induction on the construction of $\varphi$.
The base case for $\varphi\in\PROP\cup\NOM$ is straightforward,
as is the inductive step for $\varphi=\alpha\vee\beta$ and $\varphi=\alpha\wedge\beta$,
and even for $\varphi=\at_x\alpha$.
It remains to consider the case $\varphi=\Box\alpha$.
If $\eta_0,0\models\varphi$, then $\varphi\in \MSaT[\N](\Box,\at)$.
If $\Box\alpha\in \MSaT[\N](\Box,\at)$, then 
there exists $k$ such that $\eta,k\models (\Box\alpha)'$ 
(for some $\eta$, by the claim above).
Let $\alpha^{\ast}$ be the formula with $(\Box\alpha)'=\Box(\alpha^{\ast})$.
By the semantics of $\Box$ we obtain that 
there exists $k$ such that for all $k'>k$ holds $\eta,k'\models \alpha^{\ast}$.
By inductive hypothesis follows $\exists k \forall k'>k: \eta_0,0\models \alpha^{\ast}$,
what is equivalent to $\eta_0,0\models \alpha^{\ast}$.
Notice that $\alpha^{\ast}$ contains no nominal.
By the monotonicity of $\alpha$,
it follows that for all $k\in\Nat$ holds $\eta_0,k\models \alpha^{\ast}$.
When we re-replace the $\bot$'s by the replaced nominals,
the satisfaction is kept because of the monotonicity of $\alpha$, 
and therefore for all $k\in\Nat$ holds $\eta_0,k\models \alpha$.
This implies $\eta_0,0\models \Box\alpha$, which eventually yields $\varphi\in \MSaT[\N](\Box,\at)$.
\end{proofofclaim}

Both claims together yield that, 
in order to decide $\varphi\in \MSaT[\N](\Box,\at)$,
it suffices to check whether $\eta_0,0\models\varphi'$.
No nominal in $\varphi'$ occurs bound by a $\Box$-operator.
Therefore for every subformula $\Box\alpha$ of $\varphi'$  and for every $k$ holds:
$\eta_0,k\models\alpha$ if and only if $\eta_0,0\models\alpha$.
All nominals that occur free or bound by an $\at$ 
evaluate to \emph{true} in state $0$ via $\eta_0$.
Therefore, in order to decide $\eta_0,0\models\varphi'$,
it suffices to ignore all $\Box$ and $\at$-operators of $\varphi'$
and evaluate it as a propositional formula under assignment $\eta_0$
that sets all atoms of $\varphi'$ to \emph{true}.
This can be done in $\NC1$ \cite{bus87}.
\end{proof}

\section{Proof of Theorem \ref{thm:NC1-compl-BD}}
\label{app:NC1_proof_BD}

\textbf{Theorem \ref{thm:NC1-compl-BD}}
\emph{$\MSaT[\N](\Box,\dna)$ is $\NC1$-complete.}

\begin{proof}
$\NC1$-hardness follows from Theorem~\ref{thm:all-M-NC1-compl}.

For the upper bound, we distinguish occurrences of state variables as the occurrences in the proof sketch above.
They are either \emph{free}, or they are \emph{bound} by a $\Box$, or they are \emph{bound} by $\dna$.
Note that this phrasing differs from the standard usage of the terms `free' and `bound' in the context of state variables.
A free occurrence of $i$ in $\alpha$ is bound by $\Box$ in $\Box\alpha$, as above.
It is bound by $\dna$ in $\dna i . \alpha$ only.
Notice that $y$ occurs free in $\dna x . y$ (for $x\not=y$).

\begin{claim}
Let $\alpha'$ be the formula obtained from $\alpha$
by replacing every occurrence of a state variable that is bound by $\Box$ with $\bot$,
and let $g$ be an assignment.
If $g,k\models\alpha$, then $g,k\models\alpha'$.
\end{claim}

\begin{proofofclaim}
We use induction on the construction of $\varphi$.
The base case for $\varphi\in\SVAR$ is straightforward,
as is the inductive step for $\varphi=\alpha\vee\beta$, $\varphi=\alpha\wedge\beta$,
and for $\varphi=\dna x.\alpha$.
It remains to consider the case $\varphi=\Box\alpha$.
Let $g,k\models\Box\alpha$ for $k \in \Nat$.
Then for all $k'>k$: $g,k'\models \alpha$ (by semantics of $\Box$)
and by inductive hypothesis follows for all $k'>k$: $g,k'\models \alpha'$.
Assume that in $\Box(\alpha')$ there occurs a state variable $i$ that is bound by the initial $\Box$-operator.
Since for all $k'>k$ holds $g,k'\models \alpha'$,
there is some $\ell>\max \bigcup_{x\in\SVAR}g(x)$ such that $g,\ell\models \alpha'$.
Therefore $g,\ell\models \alpha'[i/\bot]$, and by the monotonicity of $\alpha'$
it follows that for all $k'>k$ holds $g,k'\models \alpha'[i/\bot]$.
In this way, all state variables bound by the initial $\Box$-operator can be replaced by $\bot$,
and it follows that $g,k\models (\Box(\alpha'))'$, where $(\Box\alpha')'=(\Box\alpha)'$.
\end{proofofclaim}

\begin{claim}
$\varphi\in \MSaT[\N](\Box,\dna)$ if and only if $g_0,0\models\varphi$, 
for $g_0(x)=0$ for every $x \in \SVAR$.
\end{claim}

\begin{proofofclaim}
We use induction on the construction of $\varphi$.
The base case for $\varphi\in\SVAR$ is straightforward,
as is the inductive step for $\varphi=\alpha\vee\beta$, $\varphi=\alpha\wedge\beta$,
and for $\varphi=\dna x.\alpha$.
It remains to consider the case $\varphi=\Box\alpha$.

If $\Box\alpha\in \MSaT[\N](\Box,\at)$,
then there exists $k$ such that $g,k\models (\Box\alpha)'$ (for some $\eta$ and $g$).
Let $\alpha^{\ast}$ be the formula with $(\Box\alpha)'=\Box\alpha^{\ast}$.
By the semantics of $\Box$ we obtain that 
there exists $k$ such that for all $k'>k$ holds $g_0,k'\models \alpha^{\ast}$,
and therefore $\alpha^{\ast}\in \MSaT[\N](\Box,\dna)$.
By inductive hypothesis follows $g_0,0\models \alpha^{\ast}$.
Notice that $\alpha^{\ast}$ contains no free state variable.
Therefore for all $k\in\Nat$ holds $g_0,k\models \alpha^{\ast}$.
When we re-replace the $\bot$'s by the replaced state variables,
the satisfaction is kept, and therefore for all $k\in\Nat$ holds $g_0,k\models \alpha$,
which eventually implies $g_0,0\models \Box\alpha$, i.e. $g_0,0\models\varphi$.
\end{proofofclaim}

Both claims together yield that
in order to decide $\varphi\in \MSaT[\N](\Box,\dna)$,
it suffices to check whether $g_0,0\models\varphi'$.
No state variable in $\varphi'$ occurs bound by a $\Box$-operator.
Therefore for every subformula $\Box\alpha$ of $\varphi'$  and for every $k$ holds:
$g_0,k\models\alpha$ if and only if $g_0,0\models\alpha$.
All occurrences of state variables in $\varphi'$ that are bound by $\dna$
evaluate to \emph{true}, because no $\Box$ occurs ``between'' the 
binding $\dna i$ and the occurrence of $i$,
which means that the state where the variable is bound is the same as where the variable is used.
All free occurrences of state variables  
evaluate to \emph{true} in state $0$ due to $g_0$.
Therefore, in order to decide $g_0,0\models\varphi'$,
it suffices to ignore all $\Box$ and $\dna$-operators of $\varphi'$
and evaluate it as a propositional formula under an assignment
that sets all atoms to \emph{true}.
This can be done in $\NC1$ \cite{bus87}.
\end{proof}

\section{Proof of Theorem \ref{thm:N-Bda-in-L}}
\label{app:N-Bda-in-L}

\textbf{Theorem \ref{thm:N-Bda-in-L}}
\emph{$\MSaT[\Nat](\Box,\dna,\at)$ is in \LOGSPACE.}

\vspace{2ex}
\noindent
For this upper bound, we will establish a characterisation of the satisfaction relation
that assigns a \emph{unique} assignment and state of evaluation to every subformula
of a given formula $\varphi$. Using this new characterisation, we will devise a
decision procedure that runs in logarithmic space and consists of two steps:
it replaces every occurrence of any state variable $x$ in $\varphi$ with
1 if its state of evaluation agrees with that of its $\dna x$-superformula,
and with 0 otherwise; it then removes all $\Box$-, $\dna$- and $\at$-operators
from the formula and tests whether the resulting Boolean formula is valid.

In what follows, we want to restrict assignments to the finitely many free state variables occurring free in a given formula $\varphi$.
For this purpose, we define the notion of a \emph{partial} assignment $g : V \to \Nat$ \emph{for} $\varphi$
where $V$ is a finite set of state variables with $\FREE_\varphi \subseteq V$,
i.e., $g$ is defined for all state variables free in $\varphi$.
Here we include subscripts of the $\at$-operator in the notion
of a free state variable: for example, $\dna x.\at_x\at_yz$ has free state variables $y,z$.
The satisfaction relation $\models$ for partial assignments is analogously defined to the definition in Section \ref{sec:prelims}.
For a partial assignment $g$ for $\dna x.\alpha$ and $i \in \Nat$, it holds that $g,i \models \dna x.\alpha$ iff $g^x_i,i \models \alpha$.
Clearly, if $g$ is a partial assignment for $\dna x.\alpha$, then $g^x_i$ is one for $\alpha$.

The definition of the satisfaction relation implies that
the satisfaction of $\Box\alpha$ at $g,i$ depends on the satisfaction of $\Box\alpha$
at infinitely many states (natural numbers) in $g$.
However, we will now show that the latter can be reduced
to satisfaction in the smallest natural number to which $g$ does not bind any state variable.
This will later imply that satisfiability of a given formula $\varphi$
can be tested by evaluating its subformulas in their \emph{uniquely determined} states $g,i$ of evaluation.

Given a partial assignment $g : V \to \mathbb{N}$, define
$ n_g = \max\{g(x) \mid x \in V\} + 1. $

\begin{lemma}
  \label{lem:box_lemma_Bda}
  For every $\varphi \in \MHL(\Box,\dna,\at)$,
  every partial assignment $g$ for $\varphi$
  and every $i \in \mathbb{N}$, it holds that
  $g,i \models \Box\varphi$ iff $g,n_g \models \varphi$.
\end{lemma}

We will prove this lemma later, using the following lemma.

\begin{lemma}
  \label{lem:box_lemma_Bda_aux}
  Let $\varphi \in \MHL(\Box,\dna,\at)$,
  let $i,j \in \mathbb{N}$,
  and let $g,h$ be partial assignments for $\varphi$
  that satisfy the following two conditions:
  \begin{enumerate}
    \item
      \label{it:box_lemma_Bda_aux_1}
      $g^{-1}(i) \subseteq h^{-1}(j)$.\\
      \quad (All state variables free in $\varphi$ and bound to $i$ by $g$ are bound to $j$ by $h$.)
    \item
      \label{it:box_lemma_Bda_aux_2}
      For all $a,b \in \FREE_\varphi$: if $g(a) = g(b)$, then $h(a) = h(b)$.\\
      \quad (Whenever $g$ binds two state variables free in $\varphi$ to one and the same state,
      so does $h$.)
  \end{enumerate}
  Then $g,i \models \varphi$ implies $h,j \models \varphi$.
\end{lemma}

\begin{proof}
  We proceed by induction on $\varphi$.
  In the base case $\varphi \in \SVAR$,
  we obtain the desired implication directly from \eqref{it:box_lemma_Bda_aux_1}.
  For the induction step, we distinguish between the possible cases for the outermost
  operator of $\varphi$.
  The Boolean cases are straightforward; the other cases are dealt with as follows.
  \begin{itemize}
    \item 
      In case $\varphi = \Box\psi$, the following chain of (bi-)implications holds.
      \begin{align*}
        g,i \models \Box\psi
        & ~\Leftrightarrow~ \forall i' > i : g, i' \models \psi \\
        & ~\Rightarrow~ g,n_g \models \psi \\
        & ~\Rightarrow~ h,n_h \models \psi \\
        & ~\Rightarrow~ \forall j' \in \mathbb{N} : h, j' \models \psi \\
        & ~\Rightarrow~ \forall j' > j : h, j' \models \psi \\
        & ~\Leftrightarrow~ h,j \models \Box\psi
      \end{align*}
      The first ``$\Rightarrow$'' is immediate in case $i < n_g$.
      Otherwise, if $i \geqslant n_g$\,, observe that $g^{-1}(\mbox{i+1}) = \emptyset = g^{-1}(n_g)$.
      Hence we can apply the induction hypothesis (IH) to $\psi,\mbox{i+1},n_g,g,h$
      because $g$ is also a partial assignment for $\psi$, the assumption \eqref{it:box_lemma_Bda_aux_1} of the IH is satisfied,
      and \eqref{it:box_lemma_Bda_aux_2} follows from the assumption \eqref{it:box_lemma_Bda_aux_2} for $\varphi,i,j,g,h$.

      The second ``$\Rightarrow$'' is due to the IH applied to
      $\psi,n_g,n_h,g,h$.
      Its assumption \eqref{it:box_lemma_Bda_aux_1} is satisfied
      because $g^{-1}(n_g) = \emptyset = h^{-1}(n_h)$,
      and \eqref{it:box_lemma_Bda_aux_2} follows from the assumption \eqref{it:box_lemma_Bda_aux_2} for $\varphi,i,j,g,h$.

      The third ``$\Rightarrow$'' is due to the IH applied to
      $\psi,n_h,j,h,h$.
      Its assumption \eqref{it:box_lemma_Bda_aux_1} is satisfied
      because $h^{-1}(n_h) = \emptyset = h^{-1}(j)$,
      and \eqref{it:box_lemma_Bda_aux_2} is obvious because $h=h$.
    \item
      In case $\varphi = \dna x.\psi$, the following chain of (bi-)implications holds.
      \begin{align*}
        g,i \models \dna x.\psi
        & ~\Leftrightarrow~ g^x_i,i \models \psi \\
        & ~\Rightarrow~ h^x_j,j \models \psi \\
        & ~\Leftrightarrow~ h,j \models \dna x.\psi
      \end{align*}
      The implication in the middle is obtained by observing that
      $g^x_i,h^x_i$ are partial assignments for $\psi$ because $g,h$ are partial assignments for $\varphi$, and applying the IH to
      $\psi,i,j,g^x_i,h^x_j$.
      Its assumption \eqref{it:box_lemma_Bda_aux_1} is satisfied
      because of the following chain of equalities and inclusions, whose middle step
      follows from the assumption \eqref{it:box_lemma_Bda_aux_2} for $\varphi,i,j,g,h$.
      \[
        {(g^x_i)}^{-1}(i) = g^{-1}(i) \cup \{x\}
                          \subseteq h^{-1}(i) \cup \{x\}
                          = {(h^x_j)}^{-1}(j)
      \]
      Assumption \eqref{it:box_lemma_Bda_aux_2} of the IH is satisfied
      for the following reason. Let $a,b \in \FREE_\psi$ with $g(a) = g(b)$.
      In case $a=b=x$, both
      $(g^x_i)(a) = (g^x_i)(b)$ and $(h^x_j)(a) = (h^x_j)(b)$ hold.
      In case $a=x$ and $b\neq x$,
      we have that $(g^x_i)(a) = (g^x_i)(b)$ implies $(g^x_i)(b) = i$,
      which implies $g(b) = i$ because $b\neq x$.
      This implies $h(b) = j$ due to the assumption \eqref{it:box_lemma_Bda_aux_1} for $\varphi,i,j,g,h$
      and because $b \in \FREE_\varphi$.
      Hence $(h^x_j)(a) = (h^x_j)(b)$.
      The case $a\neq x$ and $b=x$ is analogous to the previous one,
      and in case $a\neq x$ and $b\neq x$,
      we have that $(g^x_i)(a) = (g^x_i)(b)$ implies $g(a) = g(b)$,
      which implies $h(a) = h(b)$ due to the assumption \eqref{it:box_lemma_Bda_aux_2} for $\varphi,i,j,g,h$.
      Hence $(h^x_j)(a) = (h^x_j)(b)$.
    \item
      In case $\varphi = \at_x.\psi$, the following chain of (bi-)implications holds.
      \begin{align*}
        g,i \models \at_x\psi
        & ~\Leftrightarrow~ g,g(x) \models \psi \\
        & ~\Rightarrow~ h,h(x) \models \psi \\
        & ~\Leftrightarrow~ h,j \models \at_x\psi
      \end{align*}
      The implication in the middle is obtained by observing that $g,h$ are also partial assignments for $\psi$,
      and applying the IH to $\psi,g(x),h(x),g,h$.
      Its assumption \eqref{it:box_lemma_Bda_aux_1} is satisfied:
      consider $y \in g^{-1}(g(x))$.
      Then $g(x) = g(y)$, which implies $h(x) = h(y)$ due to
      the assumption \eqref{it:box_lemma_Bda_aux_2} for $\varphi,i,j,g,h$.
      Hence $y \in h^{-1}(h(x))$. This establishes $g^{-1}(g(x)) \subseteq h^{-1}(h(x))$.
      The assumption \eqref{it:box_lemma_Bda_aux_2} for the IH
      follows from the assumption \eqref{it:box_lemma_Bda_aux_2} for $\varphi,i,j,g,h$.
  \end{itemize}
\end{proof}

Before we can prove Lemma \ref{lem:box_lemma_Bda},
we observe the following consequence of Lemma \ref{lem:box_lemma_Bda_aux}.

\begin{corollary}
  \label{cor:box_lemma_Bda_cor}
  For every $\varphi \in \MHL(\Box,\dna,\at)$,
  every partial assignment $g$ for $\varphi$
  and every $i \in \mathbb{N}$ with $g^{-1}(i) = \emptyset$,
  it holds that $g,i \models \varphi$ implies $g,j \models \varphi$ for all $j \in \mathbb{N}$.
\end{corollary}

\begin{proof}
  It suffices to observe that the assumptions of Lemma \ref{lem:box_lemma_Bda_aux}
  are satisfied by $\varphi,i,j,g,g$ with $j \in \mathbb{N}$ arbitrary.
  \eqref{it:box_lemma_Bda_aux_1} follows from $g^{-1}(i) = \emptyset$,
  and \eqref{it:box_lemma_Bda_aux_2} holds trivially because $g=g$.
\end{proof}

We can now proceed to prove Lemma \ref{lem:box_lemma_Bda} ~~($\forall \varphi,g,i ~:~ g,i \models \Box\varphi \Leftrightarrow g,n_g \models \varphi$).

\begin{proof}[Proof of Lemma \ref{lem:box_lemma_Bda}]
  For the direction ``$\Rightarrow$'', assume that $g,i \models \Box\varphi$,
  i.e., for all $j > i$, it holds that $g,j \models \varphi$.
  In case $i < n_g$, the consequence $g,n_g \models \varphi$ is immediate.
  Otherwise, in case $i \geqslant n_g$,
  we conclude $g,i+1 \models \varphi$ from $g,i \models \Box\varphi$.
  Since $g^{-1}(i+1) = \emptyset$ in this case, we can use Corollary \ref{cor:box_lemma_Bda_cor}
  to conclude that $g,j \models \varphi$ for all $j \in \mathbb{N}$,
  and in particular for $j = n_g$.

  For the direction ``$\Leftarrow$'', assume that $g,n_g \models \varphi$.
  Then Corollary \ref{cor:box_lemma_Bda_cor} implies that $g,j \models \varphi$ for all $j \in \mathbb{N}$,
  and in particular for all $j > i$. Hence $g,i \models \Box\varphi$.
\end{proof}

Using Lemma \ref{lem:box_lemma_Bda}, we are now in a position to show
that every satisfiable formula is satisfied by a canonical assignment $g_0^\varphi$
in the state 0. We will furthermore use the characterisation of satisfaction for $\Box$-formulas
in Lemma \ref{lem:box_lemma_Bda} to establish
that the question whether $g_0^\varphi,0 \models \varphi$
can be reduced to checking satisfaction of $\varphi$'s subformulas in \emph{uniquely determined}
states and assignments.

Let $\varphi \in \MHL(\Box,\dna,\at)$. The \emph{canonical assignment} $g_0^\varphi$
\emph{for $\varphi$} is the partial assignment for $\varphi$ that maps all $x \in \FREE_\varphi$
to 0 and is undefined for all other state variables.

\begin{theorem}
  \label{thm:canonical_assignment_Bda}
  Let $\varphi \in \MHL(\Box,\dna,\at)$.
  Then $\varphi \in \MSaT[\Nat](\Box,\dna,\at)$ iff $g_0^\varphi, 0 \models \varphi$.
\end{theorem}

\begin{proof}
  The ``if'' direction is obvious. The converse is a consequence of the following claim.
  \begin{claim}
    For every $\varphi \in \MHL(\Box,\dna,\at)$,
    every partial assignment $g$ for $\varphi$ and every $i \in \mathbb{N}$:
    if $g,i \models \varphi$, then $g_0^\varphi, 0 \models \varphi$.
  \end{claim}

  \begin{proofofclaim}
    We proceed by induction on $\varphi$.
    The base case $\varphi = x \in \SVAR$ is true because $g_0^x,0 \models x$ holds.
    For the induction step, the Boolean cases are straightforward.
    The other cases are treated as follows.
    \begin{itemize}
      \item
        In case $\varphi = \Box\psi$, the following chain of implications holds.
        \begin{align*}
          g,i \models \Box\psi
          & ~\Rightarrow~ g,n_g \models \psi \\
          & ~\Rightarrow~ g_0^\psi,1 \models \psi \\
          & ~\Rightarrow~ g_0^\varphi,1 \models \psi \\
          & ~\Rightarrow~ g_0^\varphi,0 \models \Box\psi
        \end{align*}
        The first implication is due to Lemma \ref{lem:box_lemma_Bda},
        and the second uses Lemma \ref{lem:box_lemma_Bda_aux} for $\psi,g,g_0^\psi,n_g,1$:
        remember that $g,g_0^\psi$ are for $\psi$, and observe that the assumptions
        of Lemma \ref{lem:box_lemma_Bda_aux} are satisfied because
        $g^{-1}(n_g) = \emptyset = {(g_0^\psi)}^{-1}(1)$
        and $g_0^\psi(a) = 0 = g_0^\psi(b)$ for all $a,b \in \FREE_\psi$.
        The third implication holds because $g^\psi_0 = g^\varphi_0$,
        and the fourth uses Lemma \ref{lem:box_lemma_Bda}.
      \item
        In case $\varphi = \dna x.\psi$, the following chain of implications holds.
        \begin{align*}
          g,i \models \dna x.\psi
          & ~\Leftrightarrow~ g^x_i,i \models \psi \\
          & ~\Rightarrow~ g^\psi_0,0 \models \psi \\
          & ~\Rightarrow~ {(g^\varphi_0)}^x_0,0 \models \psi \\
          & ~\Leftrightarrow~ g^\varphi_0,0 \models \dna x.\psi
        \end{align*}
        The first ``$\Rightarrow$'' is due to the induction hypothesis,
        and the second uses $g^\psi_0 = {(g^\varphi_0)}^x_0$.
      \item
        In case $\varphi = \at_x\psi$, the following chain of implications holds.
        \begin{align*}
          g,i \models \at_x\psi
          & ~\Leftrightarrow~ g,g(x) \models \psi \\
          & ~\Rightarrow~ g^\psi_0,0 \models \psi \\
          & ~\Rightarrow~ {(g^\psi_0)}^x_0,0 \models \at_x\psi \\
          & ~\Leftrightarrow~ g^\varphi_0,0 \models \at_x\psi
        \end{align*}
        The first ``$\Rightarrow$'' is due to the induction hypothesis,
        and the second uses $g^\psi_0 = {(g^\varphi_0)}^x_0$; 
        note that $g^\psi_0 = g^\varphi_0$ does not necessarily hold because $x$
        might not be free in $\psi$.
    \end{itemize}
  \end{proofofclaim}
\end{proof}

Using Theorem \ref{thm:canonical_assignment_Bda} and Lemma \ref{lem:box_lemma_Bda},
we can now assign a unique assignment and state of evaluation to every subformula
of a given formula $\varphi$. This will lead us to characterize satisfiability of a given formula $\varphi$
by validity of the Boolean formula obtained from $\varphi$ by
(a) replacing every free state variable $x$ with 0 or 1,
depending on the compatibility between unique assignment and state of evaluation for $x$,
and (b) removing all non-Boolean operators.
After establishing this criterion, we will show that the transformation
can be achieved deterministically in logarithmic space.

Fix a formula $\varphi \in \MHL(\Box,\dna,\at)$ whose satisfiability is to be tested.
We denote subformulas of $\varphi$ as pairs $(\psi,p)$,
where $p \in \mathbb{N}$ denotes the position of $\psi$ in (the string that represents) $\varphi$.
This is necessary to distinguish between different occurrences
of the same subformula in $\varphi$.
The position of a subformula is always the position of its first character in the string representing $\varphi$.
If the subformula is $(\alpha\land\beta)$ or $(\alpha\lor\beta)$, then the position of the opening parenthesis
is relevant. Consequently, $\varphi$ has always position 0.

For a position $p$ in $\varphi$, denote by $\Next_1(p)$ and $\Next_2(p)$
the position of the immediate subformulas of the subformula at position $p$:
if the subformula of $\varphi$ at $p$ is
\begin{itemize}
  \item
    $(\alpha \lor \beta)$ or $(\alpha \land \beta)$,
    then $\Next_1(p)$ and $\Next_2(p)$ are the positions of $\alpha$ and $\beta$, respectively;
  \item
    $\Box\alpha$, $\dna x.\alpha$ or $\at_x\alpha$,
    then $\Next_1(p)$ is the position of $\alpha$, and $\Next_2(p)$ is undefined;
  \item
    is any other formula, then both $\Next_1(p)$ and $\Next_2(p)$ are undefined.
\end{itemize}

We now define a unique \emph{state of evaluation} $\SE^\varphi(\psi,p)$
for a subformula $\psi$ of $\varphi$ at position $p$ recursively on $p$ as follows.
\begin{itemize}
  \item
    $\SE^\varphi(\varphi,0) = (g_0^\varphi,0)$.
  \item
    For $\circ \in \{\land,\lor\}$,
    if $\SE^\varphi((\alpha \circ \beta), p) = (g,i)$,
    then $\SE^\varphi(\alpha, \Next_1(p)) = \SE^\varphi(\beta, \Next_2(p)) = (g,i)$.
  \item
    If $\SE^\varphi(\Box\alpha,p) = (g,i)$, then $\SE^\varphi(\alpha,\Next_1(p)) = (g,n_g)$.
  \item
    If $\SE^\varphi(\dna x.\alpha,p) = (g,i)$, then $\SE^\varphi(\alpha,\Next_1(p)) = (g^x_i,i)$.
  \item
    If $\SE^\varphi(\at_x\alpha,p) = (g,i)$, then $\SE^\varphi(\alpha,\Next_1(p)) = (g,g(x))$.
\end{itemize}
Observe that the first component in $\SE^\varphi(\psi,p)$ is always a partial assignment for $\psi$.

Now consider a subformula $(x,p)$ of $\varphi$ with $x \in \SVAR$ and $\SE^\varphi(x,p) = (g,i)$.
We define a function $\rep^\varphi$ mapping $x$ to $\top$ if $g(x) = i$
(i.e., $x$ is satisfied at $\SE^\varphi(x,p)$),
and to $\bot$ otherwise.
Using $\rep^\varphi$,
we now recursively define a function $\bool^\varphi$ mapping subformulas of $\varphi$
to Boolean formulas with only monotone operators and without propositional variables:
\begin{align*}
  \bool^\varphi(x,p)                  & = \rep^\varphi(x,p),\quad x \in \SVAR   \\
  \bool^\varphi(c,p)              & = c, \quad c \in \{\top,\bot\}  \\
  \bool^\varphi(\alpha \circ \beta,p) & = \bool^\varphi(\alpha, \Next_1(p)) \circ \bool^\varphi(\beta, \Next_2(p)),\quad
                                          \circ \in \{\land,\lor\}              \\
  \bool^\varphi(\Delta\alpha,p)        & = \bool^\varphi(\alpha, \Next_1(p)),\quad
                                          \Delta \in \{\Box,\dna x,\at_x\}
\end{align*}
Furthermore, let $\bool(\varphi) = \bool^\varphi(\varphi,0)$.

\begin{lemma}
  \label{lem:sat_Bda_via_validity_bool}
  Let $\varphi \in \MHL(\Box,\dna,\at)$.
  For all subformulas $(\psi,p)$ of $\varphi$, it holds that
  $\SE^\varphi(\psi,p) \models \psi$ iff $\bool^\varphi(\psi,p)$ is valid.
\end{lemma}

\begin{proof}
  We proceed by induction on $\psi$.
  Let $\SE^\varphi(\psi,p) = (g,i)$.
  The base case $\psi=x$ follows from the definition of $\bool^\varphi(x,p)$ and $\rep^\varphi(x,p)$.
  For the inductive step, the cases $\psi=\top,\bot$ follow from the definition of $\bool^\varphi$.
  The other cases are as follows.
  \begin{itemize}
    \item
      In case $\psi = \alpha \lor \beta$,
      we observe the following chain of equivalent statements.
      \begin{align*}
        g,i\models \alpha\lor\beta
        & ~\Leftrightarrow~ g,i \models\alpha \text{~or~} g,i \models \beta          \\
        & ~\Leftrightarrow~ \SE^\varphi(\alpha,\Next_1(p)) \models\alpha \text{~or~}
                            \SE^\varphi(\beta,\Next_1(p)) \models\beta               \\
        & ~\Leftrightarrow~ \bool^\varphi(\alpha,\Next_1(p)) \text{~is valid or~}
                            \SE^\varphi(\beta,\Next_1(p)) \text{~is valid}           \\
        & ~\Leftrightarrow~ \bool^\varphi(\alpha,\Next_1(p)) \lor
                            \SE^\varphi(\beta,\Next_1(p)) \text{~is valid}           \\
        & ~\Leftrightarrow~ \bool^\varphi(\alpha\lor\beta,p) \text{~is valid}
      \end{align*}
      The second equivalence is due to the definition of $\SE^\varphi$,
      the third uses the induction hypothesis, and the fifth
      is due to the definition of $\bool^\varphi$.
    \item
      The case $\psi = \alpha \land \beta$ is analogous.
    \item
      In case $\psi = \Box\alpha$,
      we observe the following chain of equivalent statements.
      \begin{align*}
        g,i\models \Box\alpha
        & ~\Leftrightarrow~ g,n_g \models\alpha                               \\
        & ~\Leftrightarrow~ \SE^\varphi(\alpha,\Next_1(p)) \models\alpha      \\
        & ~\Leftrightarrow~ \bool^\varphi(\alpha,\Next_1(p)) \text{~is valid} \\
        & ~\Leftrightarrow~ \bool^\varphi(\Box\alpha,p) \text{~is valid}
      \end{align*}
      The first equivalence uses Lemma \ref{lem:box_lemma_Bda},
      the second is due to the definition of $\SE^\varphi$,
      the third uses the induction hypothesis, and the fourth
      is due to the definition of $\bool^\varphi$.
    \item
      The cases $\psi = \dna x.\alpha$ and $\psi = \at_x\alpha$ are analogous to the previous one,
      but with the first equivalence via the definition of satisfaction.
  \end{itemize}
\end{proof}

\begin{theorem}
  \label{thm:sat_Bda_via_validity_bool}
  Let $\varphi \in \MHL(\Box,\dna,\at)$.
  Then $\varphi \in \MSaT[\Nat](\Box,\dna,\at)$ iff $\bool(\varphi)$ is valid.
\end{theorem}

\begin{proof}
  The following chain of equivalences holds.
  \begin{align*}
    \varphi \text{~is satisfiable}
    & ~\Leftrightarrow~ g^\varphi_0,0 \models \varphi             \\
    & ~\Leftrightarrow~ \SE^\varphi(\varphi,0) \models \varphi    \\
    & ~\Leftrightarrow~ \bool^\varphi(\varphi,0) \text{~is valid} \\
    & ~\Leftrightarrow~ \bool(\varphi,0) \text{~is valid}
  \end{align*}
  The first equivalence follows from Theorem \ref{thm:canonical_assignment_Bda},
  the second uses the definition of $\SE^\varphi$,
  the third is due to Lemma \ref{lem:sat_Bda_via_validity_bool},
  and the fourth uses the defintion of $\bool$.
\end{proof}

The function $\bool$ is a reduction of $\MSaT[\Nat](\Box,\dna,\at)$
to the formula value problem for Boolean formulas with only monotone operators,
which is in \NC1 \cite{bus87}.
The correctness of this reduction is shown in Theorem \ref{thm:sat_Bda_via_validity_bool}.
To establish that $\MSaT[\Nat](\Box,\dna,\at) \in \LOGSPACE$,
it remains to show that $\bool(\varphi)$ can be computed in logarithmic space.
The procedure \algname{BOOL}, which will accomplish this task, will traverse its input formula $\varphi$
from left to right, and send the character $c$ read at position $p$ to the output unchanged,
unless one of the following two cases occurs.
If $c$ belongs to a $\Box$-, $\dna x.$-, or $\at_x$-operator, then $c$ is ignored.
If $c$ is a free state variable $x$,
then $\rep^\varphi(x,p)$ is computed and sent to the output instead of $c$.
Given the definition of $\bool$, $\bool^\varphi$ and $\rep^\varphi$,
this is obviously a correct decision procedure provided that $\rep^\varphi(x,p)$
is computed by a correct subroutine \algname{REP}, which we still have to describe.
The procedure \algname{BOOL} is given in Algorithm \ref{alg:bool}.

\begin{algorithm}
  \caption{Procedure \algname{BOOL}}
  \label{alg:bool}
  \begin{algorithmic}
    \REQUIRE $\varphi \in \MHL(\Box,\dna,\at)$
    \ENSURE output $\bool(\varphi)$
    \STATE $p \leftarrow 0$
    \WHILE{$p < |\varphi|$}
      \IF{an operator $\Box$, $\dna x.$ or $\at_x$ starts at position $p$}
        \STATE{$p \leftarrow \text{position immediately following that operator}$}
      \ELSIF{a state variable $x$ starts at position $p$}
        \STATE{\textbf{output} $\algname{REP}(\varphi, x, p)$}
        \STATE{$p \leftarrow \text{position immediately following~} x$}
      \ELSE
        \STATE{\textbf{output} character at position $p$}
        \STATE{$p \leftarrow p+1$}
      \ENDIF
    \ENDWHILE
  \end{algorithmic}
\end{algorithm}

To compute $\rep^\varphi(x,p)$ using the procedure \algname{REP},
we make the following crucial observation about states of evaluation.
The operators $\Box$ and $\at_x$ are \emph{jumping operators}:
$\SE^\varphi(\Box\psi,\cdot)$ and $\SE^\varphi(\psi,\cdot)$ may differ in their second component;
the same holds for $\SE^\varphi(\at_x \psi,\cdot)$ and $\SE^\varphi(\psi,\cdot)$.
Such a difference does not occur between formulas starting with
one of the other operators $\dna x.$, $\land$, $\lor$,
and their direct subformulas.
This observation can be used to compute $\rep^\varphi(x,p)$
because that value depends on the question whether there is a jumping operator
between the position $q$ where $x$ is bound and the position $p$ of $x$.
Assume that this binder $\dna x.$ leads the subformula $\dna x.\psi$,
and that $\SE^\varphi(\dna x.\psi,q) = (g,i)$ and $\SE^\varphi(x,p) = (h,j)$.
We distinguish the following cases.
\begin{description}
  \item[Case 1.]
    If there is no jumping operator between $(x,p)$ and $(\dna x.\psi,q)$,
    then it follows from the definition of $\SE^\varphi$
    that $g(x) = i$, $g(x) = h(x)$, and $i=j$ -- all three statements
    can be shown inductively on the positions in $\varphi$.
    They imply that $h(x) = j$, hence $\rep^\varphi(x,p) = \top$.
  \item[Case 2.]
    Let $\circ$ be the \emph{last} jumping operator occurring
    between positions $q$ and $p$. More precisely,
    let $r$ be the position between $q$ and $p$ such that
    \begin{itemize}
      \item
        the operator $\circ$ at position $r$ is a jumping operator,
      \item
        that operator is in the scope of $(\dna x.,q)$ and has $(x,p)$ in its scope, and
      \item
        there is no jumping operator in the scope of $(\circ,r)$ that has $(x,p)$ in its scope.
    \end{itemize}
    Let $\circ\vartheta$ be the subformula at position $r$.
    \begin{description}
      \item[Case 2.1.]
        If $\circ = \Box$, then the definition of $\SE^\varphi$ implies that
        $\SE^\varphi(\Box\vartheta,r) = (g, n_g)$ for some partial assignment $g$.
        Since $x$ is not bound between $r$ and $p$,
        and since no jumping operator occurs between $r$ and $p$,
        we conclude from the definition of $\SE^\varphi$ that $h(x) \neq n_g$ and $j = n_g$.
        Hence $h(x) \neq j$, and $\rep^\varphi(x,p) = \bot$.
      \item[Case 2.2.]
        If $\circ = \at_y$, then let $(\dna y.\eta, s)$ be
        the subformula ``above'' $\at_y\vartheta$ that binds $y$,
        with $\SE^\varphi(\dna y.\eta) = (g',i')$
        and $\SE^\varphi(\at_y.\vartheta) = (h',j')$.

        Then it holds that (a) $g(x) = h(x)$, due to the definition of $\SE^\varphi$
        and because $x$ is not bound between $q$ and $p$,
        and (b) $j=h(y)=h'(y)=g'(y)$, which follows from the definition of $\SE^\varphi$
        for $\at_y$-formulas and the fact that $y$ is not bound between $s$ and $p$.
        Therefore we have that $\rep^\varphi(x,p) = \top$ iff $g(x) = g'(x)$.
        This new criterion compares states of evaluations of subformulas at smaller positions in $\varphi$,
        and it can be decided applying the same case distinction to those two subformulas.
    \end{description}
\end{description}

We therefore obtain a recursive procedure \algname{REP} for deciding whether $\rep^\varphi = \top$.
For every recursive call according to Case 2.2,
a pair of subformulas at smaller positions in $\varphi$ is compared.
Therefore, the recursion has to terminate after at most $|\varphi|$ steps.
Since the result of a recursive call does not need to be processed any further,
\algname{REP} can be implemented using end-recursion, i.e., without a stack.
Together with the fact that only a constant number of position counters are needed
(and, consequently, determining the \emph{last} jumping operator between two positions in $\varphi$
can be implemented in logarithmic space),
Algorithm \ref{alg:rep} runs in logarithmic space.
The previous considerations imply its correctness.

\begin{figure*}[ttt!]
  \begin{algorithm}[H]
    \begin{algorithmic}
      \REQUIRE $\varphi \in \MHL(\Box,\dna,\at)$, free state variable $x$ in $\varphi$ at position $p$
      \ENSURE output $\rep^\varphi(x,p)$
      \STATE let $(\dna x.\psi,q)$ be the $\dna x.$-superformula of $\psi$ at position $q$ in $\varphi$
      \STATE \textbf{call} subroutine $\algname{REP'}(\varphi,~ (\dna x.\psi,q),~ (x,p))$
    \end{algorithmic}
    \caption{Procedure \algname{REP}}
    \label{alg:rep}
  \end{algorithm}

  \begin{algorithm}[H]
    \begin{algorithmic}
      \REQUIRE $\varphi \in \MHL(\Box,\dna,\at)$, subformulas $(\alpha,p)$, $(\beta,q)$ of $\varphi$
      \ENSURE output $\top$ if second components of $\SE^\varphi(\alpha,p)$ and $\SE^\varphi(\beta,q)$ agree, $\bot$ otherwise
      \IF{there is no jumping operator between $(\alpha,p)$ and $(\beta,q)$}
        \STATE \textbf{return} $\top$
      \ELSIF{the \emph{last} jumping operator between $(\alpha,p)$ and $(\beta,q)$ is $\Box$}
        \STATE \textbf{return} $\bot$
      \ELSIF{the \emph{last} jumping operator between $(\alpha,p)$ and $(\beta,q)$ is $\at_y$}
        \STATE let $(\dna y.\gamma, s)$ be the subformula of $\varphi$ where $y$ is bound
        \IF{$q<s$}
          \STATE \textbf{call} subroutine $\algname{REP'}(\varphi,~ (\dna x.\psi,q),~ (\dna y.\gamma, s))$
        \ELSE
          \STATE \textbf{call} subroutine $\algname{REP'}(\varphi,~ (\dna y.\gamma, s),~ (\dna x.\psi,q))$
        \ENDIF
      \ENDIF
    \end{algorithmic}
    \caption{Procedure \algname{REP'}}
    \label{alg:rep'}
  \end{algorithm}
\end{figure*}

\vspace{2ex}
\noindent
\textbf{Theorem \ref{thm:N-Bda-in-L}}
\emph{$\MSaT[\Nat](\Box,\dna,\at)$ is in \LOGSPACE.}

\section{Proof of Theorem \ref{thm:QLMP_specific}}
\label{app:QLMP_specific}

\textbf{Theorem \ref{thm:QLMP_specific}}
\emph{$\MHL(\Diamond,\Box,\at)$ has the quasi-quadratic size model property with respect to $\lin$ and $\Nat$.}

\vspace{2ex}
\noindent
We will develop a ``quasi-quadratic size model property''
for the logic $\MHL(\Diamond,\Box,\at)$ over $\lin$, and we will subsequently show how to extend the result to the other fragments
from Theorem \ref{thm:lin,N-M-NP-compl} and how to restrict them to \Nat. In the appendix, we even sketch how to obtain
an NP decision procedure for these fragments over $\lin$, $\Nat$ and the frame class $\{(\Rat,<)\}$.

Consider an arbitrary model $K = (W,<,\eta)$, and call all states in the range of $g$ \emph{nominal states}.
For every non-nominal state $w \in W$, let $\delta(w)$ be the number of states between $w$ and the next nominal state $s$.
If the next nominal state is a direct successor, then $\delta(w) = 0$; if there are infinitely many intermediary states---i.e.,
at least a part of the interval between $w$ and $s$ is dense---, then $\delta(w) = \infty$.
For every $m \geqslant 0$, we now define an equivalence relation $\equiv_m$ on $W$ as follows. $w \equiv_m w'$ if either $w=w'$
or both $w,w'$ are non-nominal states and $\delta(w) > m$ and $\delta(w') > m$.
Figure \ref{fig:thm8} gives an example for $m=3$; equivalence classes are denoted by dashed rectangles. The $i_j$ are nominal states,
and of the 8 states between $i_2$ and $i_3$, the rightmost three form separate equivalence classes, and the others form a single
equivalence class.
\begin{figure}[ht]
\begin{center}
\begin{tikzpicture}[
 scale=0.8,
 -,
 auto,
 node distance=5cm,
 label distance=3mm,
 semithick,
 state/.style={style=circle, draw=black, minimum size=4.2mm, inner sep=0mm},
 txt/.style={style=rectangle},
 eqc/.style={densely dashed,rounded corners=4pt},]
 
	\node[state] (w1) at (0,0) {$i_0$};
	\draw[eqc] (-0.5,-0.4) rectangle (0.5,0.4);
	\node[state] (w2) at (1,0) {};
	\node[state] (w3) at (2,0) {};	
	\node[state] (w4) at (3,0) {};
	\node[state] (w5) at (4,0) {};
	\draw[eqc] (0.7,-0.4) rectangle (4.4,0.4);
	\node[state] (w6) at (5,0) {};
	\draw[eqc] (5.35,-0.4) rectangle (4.65,0.4);
	\node[state] (w7) at (6,0) {};
	\draw[eqc] (5.7,-0.4) rectangle (6.3,0.4);
	\node[state] (w8) at (7,0) {};
	\draw[eqc] (6.7,-0.4) rectangle (7.3,0.4);
	\node[state] (w9) at (8,0) {$i_1$};
	\draw[eqc] (7.7,-0.4) rectangle (8.3,0.4);
	\node[state] (w10) at (10.5,0) {};
	\draw[eqc] (8.3,-0.4) rectangle (10.2,0.4);
	\draw[eqc] (10.2,-0.4) rectangle (10.8,0.4);
	\node[state] (w11) at (11.5,0) {};
	\draw[eqc] (11.2,-0.4) rectangle (11.8,0.4);
	\node[state] (w12) at (12.5,0) {$i_2$};
	\draw[eqc] (12.2,-0.4) rectangle (12.8,0.4);
	
	\path (w1) edge[->,thick] (w2)
				(w2) edge[->,thick] (w3)
				(w3) edge[->,thick] (w4)
				(w4) edge[->,thick] (w5)
				(w5) edge[->,thick] (w6)
				(w6) edge[->,thick] (w7)
				(w7) edge[->,thick] (w8)
				(w8) edge[->,thick] (w9)
				(w10) edge[->,thick] (w11)
				(w11) edge[->,thick] (w12);
	\draw [-,thick,decorate,decoration={snake,amplitude=.4mm,segment length=2mm}] (w9) -- (w10);	
	\node[txt] at (13.2,0) {$\cdots$};
	
	\node[txt] at (0,-1) {Legend:};
	\node[state] (v1) at (1.5,-1) {$w$};
	\node[state] (v2) at (3,-1) {$v$};
	\path (v1) edge[->,thick] (v2);
	\node[txt,text width=80mm] at (8.5,-1) {: ~$v$ is a \underline{direct} successor of $w$};
	\node[state] (v3) at (1.5,-1.8) {$w$};
	\node[state] (v4) at (3,-1.8) {$v$};
	\draw [-,thick,decorate,decoration={snake,amplitude=.4mm,segment length=2mm}] (v3) -- (v4);	
	\node[txt,text width=80mm] at (8.5,-1.8) {: ~$w$  and $v$ are begin and end of a \underline{dense} interval};
	\node[state] (v6) at (1.5,-2.6) {$w$};
	\node[state] (v7) at (3,-2.6) {$v$};
	\path (v6) edge[->,thick] (v7);
	\draw[eqc] (1.1,-3) rectangle (3.4,-2.2);
	\node[txt,text width=80mm] at (8.5,-2.6) {: ~$w$  and $v$ are in the same \underline{equivalence class}};

\end{tikzpicture}
\end{center}
\caption%
{An example for $m=3$}
\label{fig:thm8}
\end{figure}
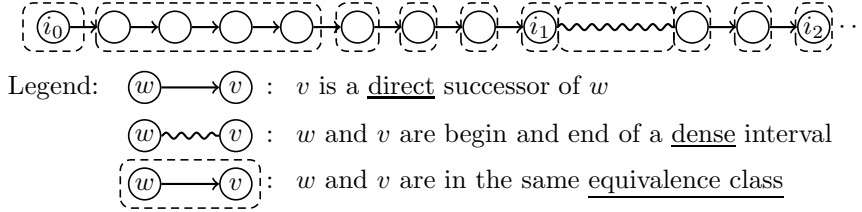 

The intuition behind this equivalence relation is that $w$ and $w'$ cannot be distinguished by formulas of modal depth $\leqslant m$.

If $w \equiv_m w'$, we call $w$ and $w'$ \emph{$m$-inseparable}, and we denote the equivalence class of $w$ w.r.t.\ $\equiv_m$
by $[w]_m$. The definition of $\equiv_m$ has the consequence $[w]_m \subseteq [w]_{m-1}$, for all $m > 0$.

It is possible to enforce dense parts in satisfying models, for instance via the following formula, which is
satisfiable in a linear structure only if that structure ends with a state satisfying the nominal $j$,
and that state needs to be the end point of a dense interval. This formula is therefore
not satisfiable over \Nat.
\[
  \varphi_d = i \wedge \Diamond\Diamond j \wedge \Box(j \vee \Diamond\Diamond j)
\]
For this reason, an equivalence class can also consist of infinitely many states. In the case of a model satisfying $\varphi_d$,
all points between $i$ and $j$ belong to the same equivalence class because all these points have an infinite distance to $j$.

The following lemma states that $m$-inseparable states cannot be distinguished by formulas of modal depth $\leqslant m$.

\begin{lemma}
  \label{lem:equiv_class_inseparable}
  For every $m \geqslant 0$,
  every formula $\varphi \in \MHL(\Diamond,\Box,\at)$ with $\md(\varphi) \leqslant m$,
  every linear model $K = (W,<,\eta)$,
  and all $w,w' \in W$ with $w \equiv_m w'$:
  \[
      K,w \models \varphi \quad\Leftrightarrow\quad K,w' \models \varphi.
  \]
\end{lemma}

\begin{proof}
  We proceed by induction on the structure of $\varphi$.
  The case for nominals is obvious because nominal states are $m$-inseparable only from themselves.
  The Boolean cases are straightforward.
  
  \begin{description}
    \item[$\varphi = \Diamond\psi$.]
      For symmetry reasons, it suffices to show ``$\Rightarrow$''.
      Let $K,w \models \Diamond\psi$ and $w \equiv_m w'$.
      Then there is some $v>w$ with $K,v \models \psi$.
      We now distinguish several cases of how $w,w',v$ are located in relation to each other.
      \begin{description}
        \item[$w'<w$.]
          Then $w<v$ implies $w'<v$, and hence $K,w' \models \Diamond\psi$.
        \item[$w\leqslant w' < v$].
          Then, still, $w'<v$, and hence $K,w' \models \Diamond\psi$.
        \item[$w<v\leqslant w'$].
          Since $w \equiv_m w'$, we have $w \equiv_m v \equiv_m w'$.
          In case $|[w]_m| < \infty$, there are exactly $m$ states between $[w]_m$ and the next nominal state.
          Let $v'$ be the $<$-least of them; then $w \equiv_{m-1} v \equiv_{m-1} w'\equiv_{m-1} v'$.
          Since $\md(\psi) = m-1$, we get $K,v' \models \psi$ via the induction hypothesis. Hence, $K,w' \models \Diamond\psi$.
          
          In case $|[w]_m| = \infty$, we conclude that at least a subinterval of $[w]_m$ is dense, and therefore $w'$ has a successor $v'$ in
          $[w]_m \subseteq [w]_{m-1}$. We can continue the argument as in the previous case.
      \end{description}
    \item[$\varphi = \Box\psi$.]
      As above, it suffices to show ``$\Rightarrow$''.
      Let $K,w \models \Box\psi$ and $w \equiv_m w'$.
      Then, for all $v>w$, we have that $K,v \models \psi$.
      Again, we consider the two cases $|[w]_m| < \infty$ and $|[w]_m| = \infty$, and fix the same $v'$ as above.
      Since $v'$ is $(m-1)$-inseparable from $w$ and $w'$, $\psi$ is also satisfied by all states in
      $[w]_m$. Therefore, $K,v \models \psi$ for all $v > w'$, hence, $K,w' \models \Box\psi$.
    \item[$\varphi = \at_i\psi$.]
      Then $K,w \models \at_i\psi ~\Leftrightarrow~ K,v \models \psi \text{~for any~} v ~\Leftrightarrow~ K,w' \models \at_i\psi$.
  \end{description}
\end{proof}

We now use this inseparability result to reduce a satisfying model in size such that it can be represented
in polynomial space. Fix a formula $\varphi$ with $\md(\varphi) = m$ and a linear model $K$ with $K,w \models \varphi$ for some state $w$.
If it were not possible to enforce dense intervals, it would suffice to collapse every $m$-equivalence class of $K$ to a single point,
i.e., the quotient model of $K$ w.r.t.\ $\equiv_m$ would satisfy $\varphi$ at $[w]_m$. This would serve our purpose over \Nat.
In contrast, an infinite equivalence class (IEC)---which
has to contain a dense subinterval---needs to remain dense for the next lemma to work. For a uniform representation,
we replace any IEC with a copy of $(0,1)_\Rat$, the open interval of all rationals between 0 and 1.
Since a dense interval can be of higher cardinality than $(0,1)_\Rat$---just consider \R, for example---,
we cannot expect to map every point of an IEC M to a point in the associated copy of $(0,1)_\Rat$.
Instead, we use a surjective partial morphism $f: (M,<) \to (0,1)_\Rat$, i.e., a partial function
that satisfies the equivalence $x=y \Leftrightarrow f(x) = f(y)$ for all $x,y \in M$
and whose range is all of \Rat. These conditions ensure that every $x \in \dom(f)$ has a successor $y \in \dom(f)$
with $f(x) < f(y)$. Such a function always exists: since every IEC $[w]_m$ contains a dense subinterval, it also contains
an isomorphic copy of $(0,1)_\Rat$.

The refined ``quotient'' model $K_m = (W_m, <_m, \eta_m)$ is now constructed as follows.
For every infinite $[w]_m$, let $\infimage{[w]_m}$ be a fresh copy of $(0,1)_\Rat$. We set
\begin{itemize}
  \item
    $\displaystyle W_k = \biguplus_{|[w]_m| = \infty} \infimage{[w]_m} \quad\uplus \{[w]_m : |[w]_m| < \infty\}$
  \item
    $[w]_m <_m [v]_m$ if $[w]_m$ and $[v]_m$ are finite and $w' < v'$ for some $w' \in [w]_m$ and $v' \in [v]_m$
  \item
    $q <_m q'$ if $q,q' \in \infimage{[w]_m}$ for some $w$ with $|[w]_m| = \infty$, and $q<q'$ on $(0,1)_\Rat$
  \item
    $q <_m [v]_m$ if $q \in \infimage{[w]_m}$ for some $w$ with $|[w]_m| = \infty$, $[v]_m$ is finite, and 
    $w < v'$ for some $v' \in [v]_m$
  \item
    $[w]_m <_m q'$ if $q' \in \infimage{[v]_m}$ for some $v$ with $|[v]_m| = \infty$, $[w]_m$ is finite, and 
    $w' < v$ for some $w' \in [w]_m$
  \item
    $\eta_m(i) = [\eta(i)]_m$
\end{itemize}
We also define a \emph{model reduction function} for $K$ to be a surjective partial function $f : K \to K_m$ with the following conditions.
\begin{itemize}
  \item
    If $|[w]_m| < \infty$, then $f(w') = [w]_m$ for all $w' \in [w]_m$.
  \item
    If $|[w]_m| = \infty$, then $f(w') = g(w')$ for all $w' \in [w]_m$,
    for some surjective partial morphism $g : [w]_m \to \infimage{[w]_m}$.
\end{itemize}

\begin{lemma}
  \label{lem:model_reduction}
  For every $m \geqslant 0$,
  every formula $\varphi \in \MHL(\Diamond,\Box,\at)$ with $\md(\varphi) \leqslant m$,
  every linear model $K = (W,<,\eta)$,
  every model reduction function $f$ for $K$
  and all $w \in \dom(f)$:
  \[
      K,w \models \varphi \quad\Leftrightarrow\quad K_m,f(w) \models \varphi.
  \]
\end{lemma}

\begin{proof}
  We proceed by induction on $\varphi$.
  The atomic and Boolean cases are straightforward again.
  \begin{description}
    \item[$\varphi = \Diamond\psi$.]
      Let $K,w \models \varphi$.
      \begin{description}
        \item[{Case 1: $|[w]_m| < \infty$.}]
          Let $w'$ be the $<$-greatest member of $[w]_k$.
          Due to Lemma \ref{lem:equiv_class_inseparable}, $K,w' \models \Diamond\psi$.
          Therefore there is some $v > w'$ with $K,v \models \psi$ and $v \not\equiv_m w$.
          If $|[v]_m| < \infty$, then $v \in \dom(f)$, and the induction hypothesis yields
          $K_m,f(v) \models \psi$.
          Since $w<v$ with $v \not\equiv_m w$, we obtain $f(w)<_m f(v)$, hence $K_m,f(w) \models \psi$.
          If $|[v]_m| = \infty$, we use $[v]_m \subseteq [v]_{m-1}$ and conclude from Lemma \ref{lem:equiv_class_inseparable}
          that $K,v' \models \psi$ for all $v' \in [v]_m$. Take such a $v'$ with $v' \in \dom(f)$ and apply the induction hypothesis as in the case $|[v]_m| < \infty$.
        \item[{Case 2: $|[w]_m| = \infty$.}]
          Since $K,w \models \varphi$, there is some $v > w$ with $K,v \models \psi$.
          If $v \not\equiv_m w$, then we argue as in Case 1.
          Otherwise, we use Lemma \ref{lem:equiv_class_inseparable} to conclude that
          $K,v' \models \psi$ for all $v' \in [w]_m$.
          Since the restriction of $f$ to $[w]_m$ is a surjective morphism and $(0,1)_\Rat$ is dense,
          there is some $v' > w$ with $v' \in [w]_m$, $v' \in \dom(f)$ and $f(w) <_m f(v')$.
          From $K,v' \models \psi$ we conclude via the induction hypothesis that $K_m,f(v') \models \psi$,
          hence $K_m,f(w) \models \Diamond\psi$.
      \end{description}
    \item[$\varphi = \Box\psi$.]
      Let $K,w \models \Box\psi$, i.e., $K,v \models \psi$ for all $v > w$.
      Then $K,v \models \psi$ for all $v$ with $v \in \dom(f)$ and $f(v) >_m f(w)$.
      Due to the induction hypothesis, $K_m,f(v) \models \psi$ for all $v$ with $v \in \dom(f)$ and $f(v) >_m f(w)$.
      Since $f$ is surjective, we have $K_m,v' \models \psi$ for all $v' \in W_m$ with $v' >_m f(w)$.
      Hence $K_m,f(w) \models \Box\psi$.
    \item[$\varphi = \at_i\psi$.]
      Let $K,w \models \at_i\psi$, i.e., $K,\eta(i) \models \psi$.
      Then $K_m,\eta_m(i) \models \psi$ due to the induction hypothesis and the definition of $K_m$.
      Hence $K_m,f(w) \models \at_i\psi$.
  \end{description}
\end{proof}

At this point, it is important to notice that, if $K$ is a model over \Nat,
then so is $K_m$. Therefore, Lemma \ref{lem:model_reduction} gives us a quasi-quadratic size model property
for $\MHL(\Diamond,\Box,\at)$ over $\lin$ as well as \Nat\ -- and also over $\{(\Rat,<)\}$, see appendix.
We say that a model $K$ is of \emph{size quasi-quadratic} in an integer $m$
if every interval between two consecutive nominal states in $K$
consists of at most $m$ states, possibly with one preceding isomorphic copy of $(0,1)_\Rat$.
We furthermore say that a fragment $\MHL(O)$ has the \emph{quasi-quadratic size model property} with respect to
a frame class \Fclass{F} if,
for every $\varphi \in \MSaT[\Fclass{F}](O)$,
there exists a model over a frame in \Fclass{F} that is of size quasi-quadratic in $\md(\varphi)$ and satisfies $\varphi$.

\vspace{2ex}
\noindent
\textbf{Theorem \ref{thm:QLMP_specific}}
\emph{$\MHL(\Diamond,\Box,\at)$ has the quasi-quadratic size model property with respect to $\lin$ and $\Nat$.}

\begin{proof}
  Let $K=(W,<,\eta)$ be a linear model and $w_0 \in W$ with $K,w_0 \models \varphi$.
  Consider $\varphi' = i \wedge \varphi$ for a fresh nominal $i$.
  Let $m = \md(\varphi) = \md(\varphi')$. Then $\varphi'$ is satisfiable in the $w_0$ of the model $K'$ obtained from $K$ by interpreting $i$ in $w_0$.
  Now take an arbitrary model reduction function $f$ for $K'$, which has to have $w_0$ in its domain,
  and apply Lemma \ref{lem:model_reduction} to obtain $K_m,f(w_0) \models \varphi'$.
\end{proof}

}

\end{document}